\documentclass[11pt]{article} 

\usepackage[utf8]{inputenc} 

\usepackage{cite}

\usepackage{hyperref}
\hypersetup{colorlinks=false}

\usepackage{geometry,amsmath} 
\usepackage{graphicx} 

\usepackage{lipsum}

\numberwithin{equation}{section}

\usepackage{booktabs} 
\usepackage{array} 
\usepackage{paralist} 
\usepackage{verbatim} 
\usepackage{subfig} 

\usepackage{float}

\usepackage{todonotes}

\usepackage{amssymb}
\usepackage{amsmath}
\usepackage{amsthm}
\usepackage{bm}
\usepackage{mathrsfs}

\usepackage{xcolor,pict2e}

\newcommand{\intpsym}{\,\, \setlength{\unitlength}{0.4mm}
	\begin{picture}(0, 0)(5, 5)%
	\put(0, 4){\line(1,0){5}} \put(5,4){\line(0,1){8}}
	\end{picture}\,\,}
\newcommand{\intp}{\,\,\intpsym}
\newcommand{\es}{\operatorname{S}}

\newcommand{\te}{\operatorname{T}}
\newcommand{\teJ}{\te_\mathbf{J}}
\newcommand{\teK}{\te_\mathbf{K}}
\newcommand{\id}{\operatorname{id}}

\newcommand{\volf}{\operatorname{vol}}
\newcommand{\upd}{\operatorname{d}}
\newcommand{\uph}{\operatorname{h}}
\newcommand{\upv}{\operatorname{v}}
\newcommand{\p}{\partial}
\newcommand{\dv}{\upd_{\upv}\!}
\newcommand{\diffi}{\mathrm{D}_{n^i}}
\newcommand{\diff}[1]{\mathrm{D}_{n^{#1}}}
\newcommand{\diffn}{\mathrm{D}_{n}}
\newcommand{\pde}{P$\Delta$E}

\newcommand{\mbn}{\mathbf{n}}
\newcommand{\mbu}{\mathbf{u}}
\newcommand{\mbv}{\mathbf{v}}
\newcommand{\mbx}{\mathbf{x}}
\newcommand{\mbI}{\mathbf{I}}
\newcommand{\mbJ}{\mathbf{J}}
\newcommand{\mbK}{\mathbf{K}}
\newcommand{\mcE}{\mathcal{E}}
\newcommand{\mcI}{\mathcal{I}}
\newcommand{\mcH}{\mathcal{H}}
\newcommand{\mcL}{\mathcal{L}}
\newcommand{\msL}{\mathscr{L}}
\newcommand{\bD}{\bm{\Delta}}
\newcommand{\dD}{\upd^{\vartriangle}}
\newcommand{\dDh}{\upd^{\vartriangle}_{\uph}}

\theoremstyle{plain}
\newtheorem{thm}{Theorem}[section]

\newtheorem{lem}[thm]{Lemma}
\newtheorem{prop}[thm]{Proposition}

\theoremstyle{definition}
\newtheorem{defn}{Definition}[section]
\newtheorem{exm}[defn]{Example}
\newtheorem{rem}[defn]{Remark}

\usepackage{fancyhdr} 
\usepackage{sectsty}
\allsectionsfont{\sffamily\mdseries\upshape} 

\usepackage[nottoc,notlof,notlot]{tocbibind} 
\usepackage[titles,subfigure]{tocloft} 

\title{The difference variational bicomplex and multisymplectic systems}
\author{Linyu Peng$^{1}$\footnote{Corresponding author. Email: l.peng@mech.keio.ac.jp} ~ and Peter E. Hydon{$^2$}\footnote{Email: p.e.hydon@kent.ac.uk} \vspace{0.4cm}
\\
{\it 1. Department of Mechanical Engineering, Keio University,} \\
{\it Yokohama 223-8522, Japan}\\
{\it 2. School of Mathematics, Statistics and Actuarial Science,} \\
{\it University of Kent,  Canterbury CT2 7NF, UK}\\ }

\begin{document}

\maketitle

\begin{abstract}
The difference variational bicomplex, which is the natural setting for systems of difference equations, is constructed and used to examine the geometric and algebraic properties of various systems. Exactness of the bicomplex gives a coordinate-free setting for finite difference variational problems, Euler--Lagrange equations and Noether's theorem. We also examine the connection between the condition for the existence of a Hamiltonian and the multisymplecticity of systems of partial difference equations. Furthermore, we define difference multimomentum maps of multisymplectic systems, which yield their conservation laws.
To conclude, we adapt the variational bicomplex to multisymplectic integrators on a mesh that is logically rectangular. By scaling horizontal forms and difference operators according to the local step sizes, all of the results derived earlier can be applied, whether or not the mesh is uniform.
\vspace{0.2cm}

{Keywords: difference variational bicomplex; multisymplectic
system; multisymplectic integrator; conservation law; multimomentum map}\vspace{0.2cm}

{Mathematics Subject Classification: 65P10;  14F43;  39A12}

\end{abstract}

\section{Introduction}
Symmetry methods provide powerful tools for obtaining solutions and conservation laws of a given system of partial differential equations (PDEs) and for understanding structural features such as integrability \cite{BCA2010,Hy2000b,Ol1993}. In the formal geometric approach, the variational complex is central to the study
of symmetries, scalar conservation laws and Euler--Lagrange equations \cite{KO2003,Ol1993}. 
This complex is contained in the (augmented) variational bicomplex, which is a natural geometric setting for all of the above and also for multisymplectic PDEs \cite{BrHyLa2010} and other PDEs with form-valued conservation laws, as well as the Lagrangian multiform (or pluri-Lagrangian) formalism of integrable systems (see, e.g., \cite{LN2009,SV2016,SNC2020}).

The (free) variational bicomplex is constructed by splitting the exterior
derivative into horizontal and vertical parts (see \cite{Vi2001,KrVi1999,Vi1984,An1992,An1989}), which reflect the distinction between independent and dependent variables. It is augmented by a projection, the interior Euler operator, which is used to derive Euler--Lagrange equations from a given Lagrangian form. Independently, Anderson \cite{An1992,An1989}, Tsujishita \cite{Ts1982} and Vinogradov \cite{Vi2001,Vi1984} proved that the augmented variational bicomplex is exact.

Over the last three decades, symmetry analysis for differential equations has been extended to difference equations (see \cite{Do2001,Do2010,Hy2005,RaHy2007,PH2022,Xe2018,Hy2014,LeTrWi2000,Pe2013,Pe2017}).
Difference forms \cite{MaHy2008} and the difference variational complex
\cite{HyMa2004} have also been developed. These results are
fundamentally important for the geometric analysis of finite difference numerical schemes. For instance, a variational integrator yields Euler--Lagrange equations from a difference Lagrangian form.

This paper describes the difference variational bicomplex, as proposed in the thesis \cite{Pe2013},  and examines some of its applications. Section \ref{sec:vb} begins with a brief review of the differential variational bicomplex, which is a natural setting for multisymplectic systems of PDEs. Section \ref{sec:dvb} develops the main ideas and structures for the difference variational bicomplex. Exactness of the difference variational bicomplex plays an essential role in applications, giving coordinate-free versions of Noether's theorem  for finite difference variational problems and multisymplectic partial difference equations (P$\Delta$Es). (See the Appendix for a proof of exactness.) Section \ref{sec:hasy} uses the bicomplex to develop a coordinate-free difference version of Hamilton's
principle. In Section \ref{sec:mupde}, the
conservation of multisymplectic structures is studied for P$\Delta$Es, and we explain the link between multisymplectic systems and first-order quasilinear difference Lagrangian structures. Section \ref{sec:mommap} defines discrete
multimomentum maps, which generate conservation laws of a given multisymplectic system. In Section \ref{sec:multisyin}, we adapt the difference variational bicomplex to multisymplectic integrators on a logically-rectangular mesh.

\section{The (differential) variational bicomplex}
\label{sec:vb}
The variational bicomplex is a double complex of differential forms that arises by regarding independent variables as coordinates on a base space, with dependent variables and their derivatives coordinatizing fibres over each point of the base space. The geometric setting is the infinite prolongation bundle. We review the differential variational bicomplex (following Anderson's presentation \cite{An1992,An1989}, see also \cite{Ts1982,KrVi1999,Vi1984}), and show that this is a natural setting for multisymplectic systems of PDEs. 

\subsection{An overview of the variational bicomplex}
For differential equations with independent variables $\mbx=(x^1,x^2,\ldots,x^p)\in X\subset \mathbb{R}^p$ and dependent variables $\mbu=(u^1,u^2,\ldots, u^q)\in U\subset \mathbb{R}^q$, a natural geometric structure is the trivial fibred manifold 
\begin{equation}
\begin{aligned}
\pi: X\times U&\rightarrow X,\\
(\mbx,\mbu)&\mapsto \mbx.
\end{aligned}
\end{equation}
A solution $\mbu=f(\mbx)$ can be regarded as a local section, $s(\mbx)=(\mbx,f(\mbx))$. Restricting attention to neighbourhoods in which $f$ is smooth, a section can be prolonged to the infinite jet bundle $J^{\infty}(X\times U)$ whose coordinates represent derivatives:
\begin{equation}
\begin{aligned}
\pi^{\infty}: J^{\infty}(X\times U)&\rightarrow X,\\
(x^i,u^{\alpha},u_{\mathbf{1}_i}^{\alpha},\ldots,u^{\alpha}_{\mbJ},\ldots)&\mapsto \mbx,
\end{aligned}
\end{equation}
where $\mathbf{1}_i$ is the $p$-tuple whose only nonzero entry is $1$ in its $i$th place, $\mbJ=(j^1,j^2,\ldots,j^p)$ with all entries being non-negative integers, and $|\mbJ|=j^1+j^2+\cdots+j^p$.
The prolonged section $s$ has coordinates 
\begin{equation}
x^i,\quad u^\alpha=f^\alpha(\mbx),\quad u_{\mathbf{1}_i}^{\alpha}=\frac{\partial f^{\alpha}(\mbx)}{\partial x^i}\,,\quad \ldots,\quad u^{\alpha}_{\mbJ}=\frac{\partial^{|\mbJ|}f^{\alpha}(\mbx)}{(\partial x^1)^{j^1}(\partial x^2)^{j^2}\cdots  (\partial x^p)^{j^p}}\,, \quad \ldots 
\end{equation}
 In this setting, a differential equation defines a variety on the jet bundle.

The exterior derivative on $J^{\infty}(X\times U)$ can be written in terms of these local coordinates. Let $[\mbu]$ denote $\mbu$ and finitely many of its partial derivatives; then the exterior derivative of a locally smooth function $f(\mbx,[\mbu])$ is
\begin{equation}\label{eq:exde}
\upd\!f(\mbx,[\mbu])=\frac{\partial f(\mbx,[\mbu])}{\partial x^i}\upd\!x^i+\frac{\partial f(\mbx,[\mbu])}{\partial u^{\alpha}_{\mbJ}}\upd\! u^{\alpha}_{\mbJ}\,.
\end{equation}
The Einstein summation convention is used from \eqref{eq:exde} on. It is natural to split the exterior derivative using the contact forms $\upd\!u^{\alpha}_{\mbJ}-u^{\alpha}_{\mbJ+\mathbf{1}_i}\upd\!x^i$, because the pullback of each contact form  by any section $s$ is zero. This splitting gives
\begin{equation}\label{eq:mexde}
\upd\!f(\mbx,[\mbu])=\left(D_if(\mbx,[\mbu])\right)\upd\!x^i+\frac{\partial f(\mbx,[\mbu])}{\partial u^{\alpha}_{\mbJ}}\left(\upd\!u^{\alpha}_{\mbJ}-u^{\alpha}_{\mbJ+\mathbf{1}_i}\upd\!x^i\right),
\end{equation}
where
\begin{equation}
D_i=\frac{\partial}{\partial x^i}+u_{\mbJ+\mathbf{1}_i}^{\alpha}\frac{\partial}{\partial u_{\mbJ}^{\alpha}}
\end{equation}
is the total derivative with respect to $x^i$.  This splitting gives a basis for the set of differential one-forms: 
\begin{equation}\label{eq:dxdu}
\upd\!x^i,\qquad \upd\!u^{\alpha}_{\mbJ}-u^{\alpha}_{\mbJ+\mathbf{1}_i}\upd\!x^i,
\end{equation}
which is extended to a basis for the set $\Omega$ of all differential forms by using the wedge product. 
The exterior derivative splits into the horizontal derivative $\upd_{\uph}$ and the vertical derivative $\dv $, as follows:
\begin{equation}
\upd=\upd_{\uph}+\dv\, ,
\end{equation}
where
\begin{equation}\label{eq:dhdiff}
\upd_{\uph}=\upd\!x^i\wedge D_i\,,\quad \dv =\left(\upd\!u^{\alpha}_{\mbJ}-u^{\alpha}_{\mbJ+\mathbf{1}_i}\!\upd\!x^i\right)\wedge\frac{\partial}{\partial u^{\alpha}_{\mbJ}}\,.
\end{equation}
Direct calculation yields the identities 
\begin{equation}\label{eq:diffdhdv}
\upd_{\uph}^2=0,\qquad \upd_{\uph}\!\dv =-\dv \upd_{\uph},\qquad \upd_{\upv}^2=0.
\end{equation}
The contact forms $\dv u^{\alpha}_{\mbJ}=\upd\!u^{\alpha}_{\mbJ}-u^{\alpha}_{\mbJ+\mathbf{1}_i}\!\upd\!x^i$ form a basis for the set of vertical differential one-forms; the set of horizontal differential one-forms has a basis $\upd_{\uph}\!x^i=\upd\!x^i$. Furthermore,
\begin{equation}\label{vec1}
D_i\intp \upd_{\uph}\!x^j=\delta^j_i,\qquad D_i\intp \dv u^\beta_\mbK=0,\qquad \frac{\partial}{\partial u^{\alpha}_{\mbJ}}\intp \upd_{\uph}\!x^j=0,\qquad \frac{\partial}{\partial u^{\alpha}_{\mbJ}}\intp \dv u^\beta_\mbK=\delta^\beta_\alpha\delta^\mbJ_\mbK.
\end{equation}
A locally smooth vertical vector field $\mbv_0=Q^\alpha\,\p/\p u^\alpha$ on $X\times U$ can be prolonged to all orders to yield the vector field $\mbv=D_{\mbJ}Q^\alpha\,\p/\p u^\alpha_{\mbJ}$ on $J^{\infty}(X\times U)$. This can be generalized, with the same prolongation formula, by allowing each $Q^\alpha$ to depend on finitely many derivatives of $\mbu$, in which case $\mbv$ is a vertical generalized vector field on $J^{\infty}(X\times U)$. By the prolongation formula, $\mbv$ commutes with each $D_i$. 

A $(k+l)$-form  $\sigma$ on $J^{\infty}(X\times U)$ is said to be a $(k,l)$-\textit{form} if it can be written as 
\begin{equation}
\sigma=f_{i_1,\ldots,i_k;\alpha_1,\ldots,\alpha_l}^{\mbJ_1,\ldots,\mbJ_l}(\mbx,[\mbu])\upd_{\uph}\!x^{i_1}\wedge \cdots \wedge \upd_{\uph}\!x^{i_k}\wedge \dv  u_{\mbJ_1}^{\alpha_1}\wedge\cdots\wedge \dv u_{\mbJ_l}^{\alpha_l},
\end{equation}
where each $f_{i_1,\ldots,i_k;\alpha_1,\ldots,\alpha_l}^{\mbJ_1,\ldots,\mbJ_l}$ is a locally smooth function.
The Lie derivative of $\sigma$ with respect to a vector field $\mbv$, denoted $\mcL_\mbv \sigma$, may be obtained from Cartan's formula:
\begin{equation}
\mcL_\mbv \sigma=\mbv\intp\upd\;\!\!\sigma+\upd(\mbv\intp\sigma).
\end{equation}
Routine calculations, adapted to the horizontal and vertical derivatives, give the following results.
\begin{lem}
	Let $\sigma$ be a differential form on $J^{\infty}(X\times U)$. If $\mbv$ is a vertical generalized vector field on $J^{\infty}(X\times U)$ then
	\[
	\mbv\intp\upd_{\uph}\;\!\!\sigma+\upd_{\uph}(\mbv\intp\sigma)=0,
	\]
	so
	\[
	\mcL_\mbv\sigma=\mbv\intp\dv \;\!\!\sigma+\dv\,(\mbv\intp\sigma).
	\]
	Similarly,
\[
D_i\intp\dv \;\!\!\sigma+\dv\,(D_i\intp\sigma)=0,
\]
so
\[
\mcL_{D_i} \sigma=D_i\intp\upd_{\uph}\;\!\!\sigma+\upd_{\uph}(D_i\intp\sigma).
\]	
\end{lem}

Let $\Omega^{k,l}$ be the set of all $(k,l)$-forms over $J^{\infty}(X\times U)$. Then
\begin{equation}
\upd_{\uph}:\Omega^{k,l}\rightarrow \Omega^{k+1,l},\quad \dv :\Omega^{k,l}\rightarrow \Omega^{k,l+1},
\end{equation}
and the identities \eqref{eq:diffdhdv} yield a double complex called the variational bicomplex (Fig. \ref{fig:vb}). A $(k,l)$-form $\sigma$ is horizontally-closed if
$\upd_{\uph}\!\sigma=0$ and
horizontally exact if there exists a form $\tau\in\Omega^{k-1,l}$
such that
$\sigma=\upd_{\uph}\!\tau$.
Similarly, $\sigma$ is vertically-closed if
$\dv \sigma=0$ and
vertically exact if there exists $\tau\in\Omega^{k,l-1}$
such that
$\sigma=\dv \tau$.

\begin{figure}[htbp]
\centering
\begin{minipage}{14.5cm}
\setlength{\unitlength}{1cm}
\begin{picture}(15,7.5)
\put(0.03,0.1){$0$}
\put(0.3,0.23){\vector(1,0){0.8}}
\put(1.2,0.1){$\mathbb{R}$}
\put(1.5,0.23){\vector(1,0){0.5}}
\put(2.1,0.1){$\Omega^{0,0}$}
\put(2.8,0.23){\vector(1,0){1.5}}
\put(4.45,0.1){$\Omega^{1,0}$}
\put(5.15,0.23){\vector(1,0){1.5}}
\put(6.9,0.11){$\cdots$}
\put(7.5,0.23){\vector(1,0){1.5}}
\put(9.1,0.1){$\Omega^{p-1,0}$}
\put(10.2,0.23){\vector(1,0){1.5}}
\put(11.8,0.1){$\Omega^{p,0}$}
\put(12.5,0.23){\vector(1,0){1.5}}
\put(14.1,0.1){0}

\put(0.03,2.1){0}
\put(0.3,2.23){\vector(1,0){1.7}}
\put(2.1,2.1){$\Omega^{0,1}$}
\put(2.8,2.23){\vector(1,0){1.5}}
\put(4.45,2.1){$\Omega^{1,1}$}
\put(5.15,2.23){\vector(1,0){1.5}}
\put(6.9,2.11){$\cdots$}
\put(7.5,2.23){\vector(1,0){1.5}}
\put(9.1,2.1){$\Omega^{p-1,1}$}
\put(10.2,2.23){\vector(1,0){1.5}}
\put(11.8,2.1){$\Omega^{p,1}$}
\put(12.5,2.23){\vector(1,0){1.5}}
\put(14.1,2.1){0}

\put(0.03,4.1){0}
\put(0.3,4.23){\vector(1,0){1.7}}
\put(2.1,4.1){$\Omega^{0,2}$}
\put(2.8,4.23){\vector(1,0){1.5}}
\put(4.45,4.1){$\Omega^{1,2}$}
\put(5.15,4.23){\vector(1,0){1.5}}
\put(6.9,4.11){$\cdots$}
\put(7.5,4.23){\vector(1,0){1.5}}
\put(9.1,4.1){$\Omega^{p-1,2}$}
\put(10.2,4.23){\vector(1,0){1.5}}
\put(11.8,4.1){$\Omega^{p,2}$}
\put(12.5,4.23){\vector(1,0){1.5}}
\put(14.1,4.1){0}

\put(2.25,0.5){\vector(0,1){1.5}}
\put(2.25,2.5){\vector(0,1){1.5}}
\put(2.25,4.5){\vector(0,1){1.5}}
\put(2.23,6.3){\vdots}

\put(4.6,0.5){\vector(0,1){1.5}}
\put(4.6,2.5){\vector(0,1){1.5}}
\put(4.6,4.5){\vector(0,1){1.5}}
\put(4.58,6.3){\vdots}

\put(9.25,0.5){\vector(0,1){1.5}}
\put(9.25,2.5){\vector(0,1){1.5}}
\put(9.25,4.5){\vector(0,1){1.5}}
\put(9.23,6.3){\vdots}

\put(11.97,0.5){\vector(0,1){1.5}}
\put(11.97,2.5){\vector(0,1){1.5}}
\put(11.97,4.5){\vector(0,1){1.5}}
\put(11.95,6.3){\vdots}

\put(3.3,0.4){$\upd_{\uph}$}
\put(5.6,0.4){$\upd_{\uph}$}
\put(8,0.4){$\upd_{\uph}$}
\put(10.7,0.4){$\upd_{\uph}$}
\put(13.1,0.4){$\upd_{\uph}$}

\put(3.3,2.4){$\upd_{\uph}$}
\put(5.6,2.4){$\upd_{\uph}$}
\put(8,2.4){$\upd_{\uph}$}
\put(10.7,2.4){$\upd_{\uph}$}
\put(13.1,2.4){$\upd_{\uph}$}

\put(3.3,4.4){$\upd_{\uph}$}
\put(5.6,4.4){$\upd_{\uph}$}
\put(8,4.4){$\upd_{\uph}$}
\put(10.7,4.4){$\upd_{\uph}$}
\put(13.1,4.4){$\upd_{\uph}$}

\put(2.35,1.1){$\dv $}
\put(2.35,3.1){$\dv $}
\put(2.35,5.1){$\dv $}

\put(4.7,1.1){$\dv $}
\put(4.7,3.1){$\dv $}
\put(4.7,5.1){$\dv $}

\put(9.37,1.1){$\dv $}
\put(9.37,3.1){$\dv $}
\put(9.37,5.1){$\dv $}

\put(12.1,1.1){$\dv $}
\put(12.1,3.1){$\dv $}
\put(12.1,5.1){$\dv $}
\end{picture}
\end{minipage}
\caption{The variational bicomplex.}
\label{fig:vb}
\end{figure}

The cohomology of the variational bicomplex in Fig. \ref{fig:vb} has been well-studied; for proofs, see  \cite{An1992,Ts1982,KrVi1999,Vi1984}.
Each column of the bicomplex is the analogue of the de Rham complex
for a topologically trivial space, so that any vertically
closed form $\sigma$ is also vertically exact. However, this is not true for the rows, where the Poincar\'{e} Lemma fails. Specifically, for any $l\geq1$, there
exist horizontally-closed $(p,l)$-forms that are not horizontally
exact. To overcome this inconvenience, a projection $\mcI$ on
$\Omega^{p,l},~l\geq1$ is used to make the rows exact, yielding the augmented variational bicomplex in Fig. \ref{fig:mvb}; here $\mathscr{F}^l:=\mcI\left(\Omega^{p,l}\right)$. The cohomology groups of the augmented  variational bicomplex are all trivial, reflecting the topological triviality of $J^\infty(X\times U)$. If one pulls back the complex to the submanifold of (infinitely prolonged) solutions of a given system of PDEs, the vertical cohomology groups remain trivial, but some horizontal cohomology groups may be nontrivial.

\begin{figure}[htbp]
\centering
\begin{minipage}{15.5cm}
\setlength{\unitlength}{1.1cm}
\begin{picture}(18,7.5)
\put(-0.19,0.1){$0$} \put(0.,0.23){\vector(1,0){0.4}}
\put(0.5,0.1){$\mathbb{R}$} \put(0.8,0.23){\vector(1,0){0.4}}
\put(1.3,0.1){$\Omega^{0,0}$} \put(2.0,0.23){\vector(1,0){1.5}}
\put(3.65,0.1){$\Omega^{1,0}$} \put(4.35,0.23){\vector(1,0){1.5}}
\put(6.1,0.11){$\cdots$} \put(6.7,0.23){\vector(1,0){1.5}}
\put(8.3,0.1){$\Omega^{p-1,0}$} \put(9.4,0.23){\vector(1,0){1.5}}
\put(11.0,0.1){$\Omega^{p,0}$}

\put(-0.19,2.1){0} \put(0.,2.23){\vector(1,0){1.2}}
\put(1.3,2.1){$\Omega^{0,1}$} \put(2.,2.23){\vector(1,0){1.5}}
\put(3.65,2.1){$\Omega^{1,1}$} \put(4.35,2.23){\vector(1,0){1.5}}
\put(6.1,2.11){$\cdots$} \put(6.7,2.23){\vector(1,0){1.5}}
\put(8.3,2.1){$\Omega^{p-1,1}$} \put(9.4,2.23){\vector(1,0){1.5}}
\put(11.0,2.1){$\Omega^{p,1}$} \put(11.7,2.23){\vector(1,0){0.7}}
\put(12.5,2.1){$\mathscr{F}^1$} \put(13.03,2.23){\vector(1,0){0.8}}
\put(13.9,2.1){0}

\put(-0.19,4.1){0} \put(0.,4.23){\vector(1,0){1.2}}
\put(1.3,4.1){$\Omega^{0,2}$} \put(2.0,4.23){\vector(1,0){1.5}}
\put(3.65,4.1){$\Omega^{1,2}$} \put(4.35,4.23){\vector(1,0){1.5}}
\put(6.1,4.11){$\cdots$} \put(6.7,4.23){\vector(1,0){1.5}}
\put(8.3,4.1){$\Omega^{p-1,2}$} \put(9.4,4.23){\vector(1,0){1.5}}
\put(11.0,4.1){$\Omega^{p,2}$} \put(11.7,4.23){\vector(1,0){0.7}}
\put(12.5,4.1){$\mathscr{F}^2$} \put(13.03,4.23){\vector(1,0){0.8}}
\put(13.9,4.1){0}

\put(1.45,0.5){\vector(0,1){1.5}} \put(1.45,2.5){\vector(0,1){1.5}}
\put(1.45,4.5){\vector(0,1){1.5}} \put(1.43,6.3){\vdots}

\put(3.8,0.5){\vector(0,1){1.5}} \put(3.8,2.5){\vector(0,1){1.5}}
\put(3.8,4.5){\vector(0,1){1.5}} \put(3.78,6.3){\vdots}

\put(8.45,0.5){\vector(0,1){1.5}} \put(8.45,2.5){\vector(0,1){1.5}}
\put(8.45,4.5){\vector(0,1){1.5}} \put(8.43,6.3){\vdots}

\put(11.17,0.5){\vector(0,1){1.5}}
\put(11.17,2.5){\vector(0,1){1.5}}
\put(11.17,4.5){\vector(0,1){1.5}} \put(11.15,6.3){\vdots}

\put(2.5,0.4){$\upd_{\uph}$}
\put(4.8,0.4){$\upd_{\uph}$}
\put(7.2,0.4){$\upd_{\uph}$}
\put(9.9,0.4){$\upd_{\uph}$}

\put(2.5,2.4){$\upd_{\uph}$}
\put(4.8,2.4){$\upd_{\uph}$}
\put(7.2,2.4){$\upd_{\uph}$}
\put(9.9,2.4){$\upd_{\uph}$}
\put(11.9,2.4){$\mcI$}

\put(2.5,4.4){$\upd_{\uph}$}
\put(4.8,4.4){$\upd_{\uph}$}
\put(7.2,4.4){$\upd_{\uph}$}
\put(9.9,4.4){$\upd_{\uph}$}
\put(11.9,4.4){$\mcI$}

\put(1.55,1.1){$\dv $}
\put(1.55,3.1){$\dv $}
\put(1.55,5.1){$\dv $}

\put(3.9,1.1){$\dv $}
\put(3.9,3.1){$\dv $}
\put(3.9,5.1){$\dv $}

\put(8.57,1.1){$\dv $}
\put(8.57,3.1){$\dv $}
\put(8.57,5.1){$\dv $}

\put(11.3,1.1){$\dv $}
\put(11.3,3.1){$\dv $}
\put(11.3,5.1){$\dv $}

\put(12.75,2.5){\vector(0,1){1.5}}
\put(12.75,4.5){\vector(0,1){1.5}} \put(12.7,6.3){\vdots}

\put(12.9,3.1){$\delta_{\upv}$}
\put(12.9,5.1){$\delta_{\upv}$}

\put(11.75,0.5){\vector(2,3){0.9}}
\put(12.35,1.1){$\mcE$}
\end{picture}
\end{minipage}
\caption{The augmented variational bicomplex.}
\label{fig:mvb}
\end{figure}

To define $\mcI$, which is called the \textit{interior Euler operator}, it is helpful to use the multi-index notation 
\begin{equation}
D_{\mbJ}=D_1^{j^1}D_2^{j^2}\cdots D_p^{j^p},\qquad (-D)_{\mbJ}=(-D_1)^{j^1}(-D_2)^{j^2}\cdots (-D_p)^{j^p}=(-1)^{|\mbJ|}D_{\mbJ}\,;
\end{equation}
formally, $(-D)_{\mbJ}$ is the adjoint operator to $D_{\mbJ}$. Then $\mcI:\Omega^{p,l}\rightarrow \Omega^{p,l}$ is defined by 
\begin{equation}
\mcI(\sigma)=\frac{1}{l}\dv u^{\alpha}\wedge (-D)_{\mbJ}\left(\frac{\partial}{\partial u_{\mbJ}^{\alpha}}\intp \sigma\right),\quad \sigma\in\Omega^{p,l},
\end{equation} 
which gives the same outcome (up to a divergence) as integration by parts. The rows containing $\Omega^{k,l}$ (for fixed $l\geq 1$) are exact; in particular, $\text{ker}(\mcI)=\text{im}(\upd_{\uph})$. Moreover, the interior Euler operator is a projection (so $\mcI^2=\mcI$); hence, for each $\sigma\in\Omega^{p,l},~l\geq 1$, there exists $\tau\in\Omega^{p-1,l}$ such that 
\begin{equation}\label{eq:ii}
\sigma=\mcI(\sigma)-\upd_{\uph}\!\tau.
\end{equation}

The \textit{Euler--Lagrange operator} $\mcE:\Omega^{p,0}\rightarrow\mathscr{F}^1$ is defined by $\mcE:=\mcI\dv $. Given a Lagrangian form, $\msL[\mbu]=L(\mbx,[\mbu])\text{vol}$, where $\text{vol}=\upd\!x^1\wedge\cdots\wedge\upd\!x^p$ is the volume form, 
\begin{equation}
\mcE(\msL)=(-D)_{\mbJ}\left(\frac{\partial L(\mbx,[\mbu])}{\partial u^{\alpha}_{\mbJ}}\right)\dv u^\alpha\wedge\text{vol},
\end{equation}
so the Euler--Lagrange equations are the coefficients of $\mcE(\msL)=0$.

Bearing in mind that $\mathscr{F}^l\subset \Omega^{p,l}$, the operators $\delta_{\upv}:=\mcI\dv $ give higher variations; in particular, $\delta_{\upv}:\mathscr{F}^1\rightarrow \mathscr{F}^2$ gives the Helmholtz conditions for the inverse problem of variational calculus. 
The variational complex is the edge of the augmented variational bicomplex, consisting of the bottom row, the Euler--Lagrange operator, and the column containing the spaces $\mathscr{F}^l$. The variational complex is exact, so
\[
\text{ker}(\mcE)=\text{im}(\upd_{\uph}),\qquad \text{im}(\mcE)=\text{ker}(\delta_{\upv}),
\]
and the column containing $\mathscr{F}^l$ is exact.

\subsection{Multisymplectic systems via the variational bicomplex}

In 1997, Bridges \cite{Br1997b} introduced a multisymplectic structure that generalizes the classical finite-dimensional Hamiltonian
structure to infinite-dimensional systems. Since then, the
multisymplectic formulation and its geometry has been developed and widely applied, see 
\cite{Br1997a,BrRe2001,CoHoHy2007,HuDe2008,SuQi2003,Wa2008,Br2017,CaIbDe1999,La2000,La2004,MaPaSh1998}.

The (augmented) variational bicomplex provides a natural framework for studying multisymplectic PDEs, an approach that was introduced by Bridges \textit{et al}. \cite{BrHyLa2010}. In this framework, a system of PDEs is \textit{multisymplectic} if there exists a vertically-closed $(p-1,2)$-form $\omega$ such that $\upd_{\uph}\!\omega$ is zero on the solution submanifold (but not identically zero). For simplicity, we restrict attention to the base manifold $\mathbb{R}^p$, but the results are local and can be adapted to other base manifolds. 

The starting-point is Zuckerman's discovery in \cite{Zuck} of a `universal conserved current' for any given Lagrangian form $\msL\in \Omega^{p,0}$. This is a vertically-closed differential form $\omega\in\Omega^{p-1,2}$ that is conserved on solutions of the Euler--Lagrange equations. It is instructive to revisit Zuckerman's proof using the augmented variational bicomplex. From \eqref{eq:ii}, there exists $\eta\in\Omega^{p-1,1}$ such that
\begin{equation}\label{fund}
\mcE(\msL)=\dv \msL+\upd_{\uph}\!\eta.
\end{equation}
The $(p-1,2)$-form $\omega=\dv \eta$ satisfies
\begin{equation}\label{dom}
\upd_{\uph}\!\omega=-\dv \upd_{\uph}\!\eta=-\dv \mcE(\msL).
\end{equation}
Consequently, $\upd_{\uph}\!\omega=0$ on the solution submanifold of the Euler--Lagrange equations. So every system of Euler--Lagrange equations is multisymplectic.

Independently, Gotay \cite{Gotay} developed a covariant Hamiltonian formalism for field theories of arbitrary order, using a generalized Legendre transformation to identify additional phase space variables $p_\alpha^{\mbJ}$ from the Lagrangian. A key ingredient is the Poincar\'{e}--Cartan form $\Theta$, which is a Lepagean equivalent of the Lagrangian form. In the above notation,
\begin{equation}
\Theta=\msL+\eta,
\end{equation}
which is a $p$-form of mixed type (from the viewpoint of the bicomplex). From \eqref{fund}, we have
\begin{equation}
\upd\!\Theta=\mcE(\msL)+\omega,
\end{equation}
and \eqref{dom} arises from $\upd^2\!\Theta=0$.

Up to inessential terms, Gotay's covariant Hamiltonian formalism yields the quasilinear first-order system of Euler--Lagrange equations for the modified Lagrangian
\begin{equation}
\widehat{L}=L(\mbx,[\mbu])+p_\alpha^{\mbJ+\mathbf{1}_i}(D_iu^\alpha_\mbJ-u^\alpha_{\mbJ+\mathbf{1}_i}).
\end{equation}
Here $u^\alpha_\mbJ, p_\alpha^{\mbJ}$ are distinct dependent variables on the phase space, with the latter variables playing the role of Lagrange multipliers. So nothing is lost by restricting attention to first-order quasilinear systems, in which case $[\mbu]$ fully coordinatizes the phase space. This was the approach introduced by Bridges \cite{Br1997b}; every multisymplectic system of PDEs (in Bridges' sense) is the set of Euler--Lagrange equations for a Lagrangian of the form 
\begin{equation}\label{lag}
L(x,[\mbu])=L_{\alpha}^i(\mbx,\mbu)D_iu^{\alpha}-H(\mbx,\mbu).
\end{equation}
Specifically, the system is 
\begin{equation}
K_{\alpha\beta}^i(\mbx,\mbu)D_iu^{\beta}-\frac{\partial L_{\alpha}^i}{\partial x^i}-\frac{\partial H}{\partial
u^{\alpha}}=0,
\end{equation}
where
\begin{equation}
K^i_{\alpha\beta}(\mbx,\mbu)=\frac{\partial L_{\beta}^i}{\partial
u^{\alpha}}-\frac{\partial L_{\alpha}^i}{\partial u^{\beta}}.
\end{equation}
This gives rise to the $(p-1,2)$-form
\begin{equation}\label{omkap}
\omega=\kappa^i\wedge(D_i\intp \text{vol}),
\end{equation}
where
\begin{equation}
\kappa^i=\sum_{\alpha<\beta}K_{\alpha\beta}^i(\mbx,\mbu)\dv u^{\alpha}\wedge\dv u^{\beta}.
\end{equation}
Note that $\omega$ is vertically-closed and it does not depend on derivatives of $\mbu$. The structural conservation law,
\begin{equation}
\upd_{\uph}\!\omega=0,    
\end{equation}
holds on solutions of the system; this law is the coordinate-free version of conservation of multisymplecticity.

One can partially reverse the above derivation, using the fact that the vertical cohomology groups are trivial, even for the restricted bicomplex. As $\omega$ is vertically-closed, exactness of the vertical columns implies the existence of $\eta\in\Omega^{p-1,1}$ such that $\omega = \dv \eta$. Moreover, on the solution submanifold,
\[
\dv \upd_{\uph}\!\eta=-\upd_{\uph}\!\dv \eta=0,
\]
so there exists $\msL\in\Omega^{p,0}$ (restricted to this submanifold) such that
\[
\upd_{\uph}\!\eta=-\dv \msL.
\]
Khavkine \cite{Khavkine} used this observation to prove the existence of a Lagrangian $(p,0)$-form on the full jet space, whose Euler--Lagrange equations are solved by the given multisymplectic system of PDEs. (However, these equations may be weaker than the given system and so admit other solutions.)

It is convenient to restrict attention to multisymplectic systems that are first-order and quasilinear, so that both $\omega$ and $\eta$ are defined over each point $\mbx$ in terms of the phase space variables and their vertical derivatives. For difference equations, however, this turns out not to be possible, as differences are not defined pointwise, nor are they derivations. Nevertheless, we now describe a difference analogue of the variational bicomplex, which was introduced in the thesis \cite{Pe2013}. It has very similar features to the differential case, giving rise to a standard form for multisymplectic difference equations.

\section{Construction of the difference variational bicomplex}\label{sec:dvb}
The building-blocks for the difference variational bicomplex are difference prolongation spaces \cite{MaRoHyPe2019}, difference forms \cite{HyMa2004,MaHy2008} and the difference variational complex \cite{HyMa2004} over the base space $\mathbb{Z}^p$.

Consider a system of P$\Delta$Es with $p$ independent variables, $n^i\in\mathbb{Z}$, and $q$ dependent variables, $u^\alpha\in\mathbb{R}$. These variables can be regarded as coordinates on the \textit{total space}, $\mathbb{Z}^p\times\mathbb{R}^q$: the discrete base space $\mathbb{Z}^p$ and the connected fibres $\mathbb{R}^q$ are coordinatized respectively by $\mbn=(n^1,n^2,\ldots,n^p)$ and $\mbu=(u^1,u^2,\dots,u^q)$. For simplicity, we shall assume that this coordinate system applies everywhere, though the results below can be adapted if more than one coordinate patch is needed for a particular \pde\ system.
The fibres are mapped to one another by the horizontal translations
\begin{equation}
\begin{aligned}
\teJ:\mathbb{Z}^p\times\mathbb{R}^q&\rightarrow\mathbb{Z}^p\times\mathbb{R}^q\\
(\mbn,\mbu)&\mapsto(\mbn+\mbJ,\mbu).
\end{aligned}
\end{equation}
Note that $\teJ\circ\teK=\te_{\mbJ+\mbK}$ for all $\mbJ,\mbK\in\mathbb{Z}^p$. (See \cite{Hy2014} for other transformations of total space.)

As the total space is disconnected, it is helpful to construct a connected representation of this space over each base point. To do this, each fibre is prolonged to include the values of the coordinates on all other fibres as coordinates in a Cartesian product, using the pullback of each $u^\alpha$ with respect to every $\teJ$. The (connected) total prolongation space over an arbitrary base point, denoted $P(\mathbb{R}^q)$ (or $P_\mbn(\mathbb{R}^q)$ if the base point, $\mbn$, is specified), has coordinates $(u^\alpha_{\mbJ})$, where
\[
u^\alpha_{\mbJ}=\teJ^*\!u^\alpha.
\]
In particular, $u^{\alpha}_{\mathbf{0}}=u^{\alpha}$. The total prolongation space provides a convenient setting for the study of  geometric properties of difference equations.

The composition rule for horizontal translations gives the identities
\[
u^\alpha_{\mbJ+\mbK}=\teK^*\!u^\alpha_{\mbJ}.
\]
More generally, let $f$ be a function on $\mathbb{Z}^p\times P(\mathbb{R}^q)$ and denote its restriction to each total prolongation space $P_{\mbn}(\mathbb{R}^q)$ by $f_\mbn((u^{\alpha}_{\mbJ}))=f(\mbn,(u^{\alpha}_{\mbJ}))$; for simplicity, we assume that every $f_\mbn$ is smooth. Then the pullback of $f_{\mbn+\mbK}((u^{\alpha}_{\mbJ}))$ with respect to $\te_{\mbK}$ is the function $\te_{\mbK}^*\!f_{\mbn+\mbK}$ on $P_{\mbn}(\mathbb{R}^q)$ whose values are $f(\mbn+\mbK,(u^{\alpha}_{\mbJ+\mbK}))$. So the action of each horizontal translation $\te_{\mbK}$ on the space of smooth functions on $P_{\mbn}(\mathbb{R}^q)$ can be represented by the \textit{shift operator}
\begin{equation}
\begin{aligned}
\es_{\mbK}:C^\infty(P_\mbn(\mathbb{R}^q))&\rightarrow C^\infty(P_\mbn(\mathbb{R}^q))\\
f(\mbn,(u^{\alpha}_{\mbJ}))&\mapsto f(\mbn+\mbK,(u^{\alpha}_{\mbJ+\mbK})),
\end{aligned}
\end{equation}
so that $\es_{\mbK}\!f_\mbn=\te_{\mbK}^*\!f_{\mbn+\mbK}$. Similarly, let $\sigma$ be a differential form on $\mathbb{Z}^p\times P(\mathbb{R}^q)$ whose restriction to each $P_{\mbn}(\mathbb{R}^q)$ is $\sigma_\mbn$. Then the action of $\te_{\mbK}$ on $\sigma_\mbn$ is represented by the shift
\begin{equation}
\es_{\mbK}\!\sigma_\mbn=\te_{\mbK}^*\!\sigma_{\mbn+\mbK}.
\end{equation}
By the standard properties of the pullback, $\es_\mbK$ commutes with the wedge product and with the exterior derivative on the fibre $P(\mathbb{R}^q)$, which we denote by $\dv $ (as it acts on dependent variables only):
\begin{equation}\label{sids}
\es_\mbK(\sigma_1\wedge\sigma_2)=(\es_\mbK\!\sigma_1)\wedge(\es_\mbK\!\sigma_2),\qquad\es_\mbK(\dv \sigma)=\dv\,(\es_\mbK\!\sigma).
\end{equation}

The difference structure is a consequence of the ordering of each independent variable. For any multi-index $\mbJ=(j^1,j^2,\dots,j^p)=j^i\mathbf{1}_i$, the corresponding shift operator is
$\es_{\mbJ}=\es_{1}^{j^1}\es_{2}^{j^2}\cdots \es_{p}^{j^p}$, where $\es_i:=\es_{\mathbf{1}_i}$ denotes the forward shift with respect to $n^i$. Then the forward difference in the $n^i$-direction is represented on each $P_{\mbn}(\mathbb{R}^q)$ by the operator
\begin{equation}\label{symdiffi}
\diffi=\es_i-\id,
\end{equation}
where $\id$ is the identity mapping. In \cite{HyMa2004}, Hydon \& Mansfield introduced difference forms on $\mathbb{Z}^p$. These have the same algebraic properties as differential forms on $\mathbb{R}^p$, with the exterior algebra on $p$ symbols, $\Delta^1,\Delta^2,\dots,\Delta^p$, replacing the exterior algebra on $\upd\! x^1,\upd\! x^2,\dots,\upd\! x^p$. The symbols $\Delta^i$ at any two different points are related by (horizontal) translation, so that
\begin{equation}\label{sdel}
\Delta^i\big|_\mbn=\te_{\mbK}^*(\Delta^i\big|_{\mbn+\mbK})=:\es_\mbK(\Delta^i\big|_{\mbn}).
\end{equation}
A difference $k$-form $\sigma$ on $\mathbb{Z}^p$ assigns a $k$-form,
\[
\sigma_\mbn=f_{i_1,\ldots,i_k}(\mbn)\,
\Delta^{i_1}\big|_\mbn\wedge\cdots\wedge\Delta^{i_k}\big|_\mbn,
\]
to each $\mbn\in\mathbb{Z}^p$. In view of the invariance of $\Delta^i$ under horizontal translations, we write
\begin{equation}
\sigma=f_{i_1,\ldots,i_k}(\mbn)\,
\Delta^{i_1}\wedge\cdots\wedge\Delta^{i_k}.
\end{equation}
The exterior difference operator $\bD$ maps difference $k$-forms to difference $(k+1)$-forms as follows:
\begin{equation}
\bD\sigma=\Delta^i\wedge \diffi\sigma.
\end{equation}
Unlike the exterior derivative $\dv $, the exterior difference $\bD$ is not a derivation; however, like $\dv $, it satisfies the important identity $\bD^2=0$. Note also that $\bD n^i=\Delta^i$. The exterior difference acts pointwise on difference forms over $\mathbb{Z}^p$ and extends immediately to difference forms over $\mathbb{Z}^p\times P(\mathbb{R}^q)$,
\[
\sigma=f_{i_1,\ldots,i_k}(\mbn,(u^{\alpha}_{\mbJ}))\,
\Delta^{i_1}\wedge\cdots\wedge\Delta^{i_k}.
\]
In particular, the exterior difference of a difference $(p-1)$-form,
\[
\sigma=\sum_{i=1}^p \left((-1)^{i-1}F^i(\mbn,(u^{\alpha}_{\mbJ}))\,
\Delta^{1}\wedge\cdots\wedge\widehat{\Delta^{i}}\wedge\cdots\wedge\Delta^{p}\right)
\]
where $\widehat{\Delta^i}$ denotes the absence of $\Delta^i$, is
\[
\bD\sigma=\text{Div}\,\mathbf{F}\,\Delta^{1}\wedge\cdots\wedge\Delta^{p},\qquad\text{where}\quad \text{Div}\,\mathbf{F}:=\diffi\!\left\{ F^i(\mbn,(u^{\alpha}_{\mbJ}))\right\}.
\]
Any function of the form $\text{Div}\,\mathbf{F}$, as defined above, is called a (difference) \textit{divergence}.

To obtain the difference variational bicomplex, we combine the above exterior difference and differential structures, using the wedge product. From here on, we add the condition that any form, when restricted to a particular $P_\mbn(\mathbb{R}^q)$, depends (locally) smoothly on a finite subset of the variables $(u^{\alpha}_{\mbJ})$; such a subset is denoted $[\mbu]$. Under this condition, a $(k,l)$-form is a $(k+l)$-form, $\sigma$, that can be written (without redundancies) as
\begin{equation}\label{eq:klform}
\sigma=f^{\mbJ_1,\ldots,\mbJ_l}_{i_1,\ldots,i_k;\alpha_1,\ldots,\alpha_l}(\mbn,[\mbu])\,
\Delta^{i_1}\wedge\cdots\wedge\Delta^{i_k}\wedge\dv 
u^{\alpha_1}_{\mbJ_1}\wedge\cdots\wedge\dv 
u^{\alpha_l}_{\mbJ_l};
\end{equation}
we denote the set of all such forms by $\Omega^{k,l}$. The exterior derivative is the mapping $\dv :\Omega^{k,l}\rightarrow\Omega^{k,l+1}$ whose action on \eqref{eq:klform} gives
\begin{equation}\label{eq:dvom}
\dv \sigma=\frac{\p}{\p u^\alpha_\mbJ}\left\{f^{\mbJ_1,\ldots,\mbJ_l}_{i_1,\ldots,i_k;\alpha_1,\ldots,\alpha_l}(\mbn,[\mbu])\right\}\,
\dv u^\alpha_\mbJ\wedge\Delta^{i_1}\wedge\cdots\wedge\Delta^{i_k}\wedge\dv 
u^{\alpha_1}_{\mbJ_1}\wedge\cdots\wedge\dv 
u^{\alpha_l}_{\mbJ_l}.
\end{equation}
Shifts of \eqref{eq:klform} are given by
\begin{equation}\label{eq:siom}
\es_\mbK\;\!\!\sigma=\es_\mbK\left\{f^{\mbJ_1,\ldots,\mbJ_l}_{i_1,\ldots,i_k;\alpha_1,\ldots,\alpha_l}(\mbn,[\mbu])\right\}\,
\Delta^{i_1}\wedge\cdots\wedge\Delta^{i_k}\wedge\dv 
u^{\alpha_1}_{\mbJ_1+\mbK}\wedge\cdots\wedge\dv 
u^{\alpha_l}_{\mbJ_l+\mbK},
\end{equation}
because \eqref{sdel} implies that $\es_\mbK\!\Delta^j=\Delta^j$. 
Bearing this in mind, the exterior difference is the mapping
\begin{equation}\label{exdiff}
\begin{aligned}
\dDh:\Omega^{k,l}&\rightarrow\Omega^{k+1,l}\\
\sigma&\mapsto \Delta^i\wedge \diffi\sigma.  \end{aligned}
\end{equation}

\begin{rem}
The operator $\diffi$ is the \textit{Lie difference} with respect to the horizontal translation $\te_{\mathbf{1}_i}$, because
\[
(\diffi \sigma)|_\mbn =\te_{\mathbf{1}_i}^*(\sigma_{\mbn+\mathbf{1}_i})-\sigma_{\mbn},
\]
the right-hand side of this expression being the standard definition of the Lie difference \cite{CramPir1987}.
\end{rem}
\begin{rem}
We have used $\dDh$ instead of $\bD$ (which was designed for pure difference forms), as it is helpful to mirror the standard notation used for the differential variational bicomplex. Both $\dDh$ and $\dv $ are invariant under all  allowable changes of the coordinates used to describe their respective spaces, namely $GL(p,\mathbb{Z})$ transformations of the base space $\mathbb{Z}^p$ and diffeomorphisms of $\mathbb{R}^q$ (prolonged to $\mathbb{Z}^p\times P(\mathbb{R}^q)$). 
\end{rem}

\begin{lem}\label{lemdD}
	The operators $\dDh$ and $\dv $ satisfy the identity
	\begin{equation}\label{anticom}
	\dDh\dv =-\dv \dDh.
	\end{equation}
	Consequently, the operator $\dD:=\dDh+\dv $ satisfies $(\dD)^2=0$.
\end{lem}
\begin{proof}
	To prove \eqref{anticom}, apply $\dDh\dv $ to an arbitrary $(k,l)$-form $\sigma$, then use the identities \eqref{eq:siom} and $\es_i(\p f/\p u^\alpha_{\mbJ})=\p (\es_i\! f)/\p (\es_i\! u^\alpha_{\mbJ})$. The identity for $\dD$ is due to the symmetry of mixed partial differences (resp. derivatives) and antisymmetry of the wedge product, which gives $(\dDh)^2=0$ (resp. $\upd_{\upv}^2=0$).
\end{proof}

The operator $\dD$, which we call the \textit{exterior difference-derivative}, is analogous to the exterior derivative $\upd$ on the infinite jet bundle. It splits into horizontal and vertical components, from which the difference variational bicomplex can be constructed in the same way as for the differential case (with $\dDh$ replacing $\upd_{\uph}$).

\begin{figure}[htbp]
	\centering
	\begin{minipage}{15.5cm}
		
		\setlength{\unitlength}{1.1cm}
		\begin{picture}(18,7.5)
		\put(-0.19,0.1){$0$} \put(0.,0.23){\vector(1,0){0.4}}
		\put(0.5,0.1){$\mathbb{R}$} \put(0.8,0.23){\vector(1,0){0.4}}
		\put(1.3,0.1){$\Omega^{0,0}$} \put(2.0,0.23){\vector(1,0){1.5}}
		\put(3.65,0.1){$\Omega^{1,0}$} \put(4.35,0.23){\vector(1,0){1.5}}
		\put(6.1,0.11){$\cdots$} \put(6.7,0.23){\vector(1,0){1.5}}
		\put(8.3,0.1){$\Omega^{p-1,0}$} \put(9.4,0.23){\vector(1,0){1.5}}
		\put(11.0,0.1){$\Omega^{p,0}$}
		
		\put(-0.19,2.1){0} \put(0.,2.23){\vector(1,0){1.2}}
		\put(1.3,2.1){$\Omega^{0,1}$} \put(2.,2.23){\vector(1,0){1.5}}
		\put(3.65,2.1){$\Omega^{1,1}$} \put(4.35,2.23){\vector(1,0){1.5}}
		\put(6.1,2.11){$\cdots$} \put(6.7,2.23){\vector(1,0){1.5}}
		\put(8.3,2.1){$\Omega^{p-1,1}$} \put(9.4,2.23){\vector(1,0){1.5}}
		\put(11.0,2.1){$\Omega^{p,1}$} \put(11.7,2.23){\vector(1,0){0.7}}
		\put(12.5,2.1){$\mathscr{F}^1$} \put(13.03,2.23){\vector(1,0){0.8}}
		\put(13.9,2.1){0}
		
		\put(-0.19,4.1){0} \put(0.,4.23){\vector(1,0){1.2}}
		\put(1.3,4.1){$\Omega^{0,2}$} \put(2.0,4.23){\vector(1,0){1.5}}
		\put(3.65,4.1){$\Omega^{1,2}$} \put(4.35,4.23){\vector(1,0){1.5}}
		\put(6.1,4.11){$\cdots$} \put(6.7,4.23){\vector(1,0){1.5}}
		\put(8.3,4.1){$\Omega^{p-1,2}$} \put(9.4,4.23){\vector(1,0){1.5}}
		\put(11.0,4.1){$\Omega^{p,2}$} \put(11.7,4.23){\vector(1,0){0.7}}
		\put(12.5,4.1){$\mathscr{F}^2$} \put(13.03,4.23){\vector(1,0){0.8}}
		\put(13.9,4.1){0}
		
		\put(1.45,0.5){\vector(0,1){1.5}} \put(1.45,2.5){\vector(0,1){1.5}}
		\put(1.45,4.5){\vector(0,1){1.5}} \put(1.43,6.3){\vdots}
		
		\put(3.8,0.5){\vector(0,1){1.5}} \put(3.8,2.5){\vector(0,1){1.5}}
		\put(3.8,4.5){\vector(0,1){1.5}} \put(3.78,6.3){\vdots}
		
		\put(8.45,0.5){\vector(0,1){1.5}} \put(8.45,2.5){\vector(0,1){1.5}}
		\put(8.45,4.5){\vector(0,1){1.5}} \put(8.43,6.3){\vdots}
		
		\put(11.17,0.5){\vector(0,1){1.5}}
		\put(11.17,2.5){\vector(0,1){1.5}}
		\put(11.17,4.5){\vector(0,1){1.5}} \put(11.15,6.3){\vdots}
		
		\put(2.5,0.4){$\dDh$}
		\put(4.8,0.4){$\dDh$}
		\put(7.2,0.4){$\dDh$}
		\put(9.9,0.4){$\dDh$}
		
		\put(2.5,2.4){$\dDh$}
		\put(4.8,2.4){$\dDh$}
		\put(7.2,2.4){$\dDh$}
		\put(9.9,2.4){$\dDh$}
		\put(11.9,2.4){$\mcI^{\vartriangle}$}
		
		\put(2.5,4.4){$\dDh$}
		\put(4.8,4.4){$\dDh$}
		\put(7.2,4.4){$\dDh$}
		\put(9.9,4.4){$\dDh$}
		\put(11.9,4.4){$\mcI^{\vartriangle}$}
		
		\put(1.55,1.1){$\dv $}
		\put(1.55,3.1){$\dv $}
		\put(1.55,5.1){$\dv $}
		
		\put(3.9,1.1){$\dv $}
		\put(3.9,3.1){$\dv $}
		\put(3.9,5.1){$\dv $}
		
		\put(8.57,1.1){$\dv $}
		\put(8.57,3.1){$\dv $}
		\put(8.57,5.1){$\dv $}
		
		\put(11.3,1.1){$\dv $}
		\put(11.3,3.1){$\dv $}
		\put(11.3,5.1){$\dv $}
		
		\put(12.75,2.5){\vector(0,1){1.5}}
		\put(12.75,4.5){\vector(0,1){1.5}} \put(12.7,6.3){\vdots}
		
		\put(12.9,3.1){$\delta^{\vartriangle}_{\upv}$}
		\put(12.9,5.1){$\delta^{\vartriangle}_{\upv}$}
		
		\put(11.75,0.5){\vector(2,3){0.9}}
		\put(12.35,1.1){$\mcE^{\vartriangle}$}
		\end{picture}
	\end{minipage}
	\caption{The augmented difference variational bicomplex.}
	\label{fig:mdvb}
\end{figure}

For variational problems, a difference version of the interior Euler operator is needed to form the augmented difference variational bicomplex, which is shown in Fig. \ref{fig:mdvb}. Here, summation by parts replaces integration by parts, yielding the \textit{difference interior Euler operator} $\mcI^{\vartriangle}$ defined by
\begin{equation} 
\mcI^{\vartriangle}(\sigma):=\frac{1}{l}\,\dv u^{\alpha}
\wedge{\es_{-\mbJ}\left(\frac{\partial}{\partial
		u^{\alpha}_{\mbJ}}\intp\sigma\right)}, \qquad
\sigma\in\Omega^{p,l}.
\end{equation}
Note that $\es_{-\mbJ}$ is the formal adjoint of $\es_{\mbJ}$. From here on, let $\text{vol}=\Delta^1\wedge\cdots\wedge\Delta^p$ denote the difference version of the volume $p$-form. The \textit{difference Euler--Lagrange operator} $\mcE^{\vartriangle}:\Omega^{p,0}\rightarrow\mathscr{F}^1$ is defined by $\mcE^{\vartriangle}:=\mcI^{\vartriangle}\dv $. For a difference Lagrangian form, $\msL[\mbu]=L(\mbn,[\mbu])\text{vol}\in\Omega^{p,0}$,
\begin{equation}
\mcE^{\vartriangle}(\msL)=\es_{-\mbJ}\left(\frac{\partial L(\mbn,[\mbu])}{\partial u^{\alpha}_{\mbJ}}\right)\dv u^\alpha\wedge \text{vol}.
\end{equation}
Therefore, the difference Euler--Lagrange equations,
\[
\es_{-\mbJ}\left(\frac{\partial L(\mbn,[\mbu])}{\partial u^{\alpha}_{\mbJ}}\right)=0,
\]
are the coefficients of $\mcE^{\vartriangle}(\msL)=0$.
The operators
$\delta_{\upv}^{\vartriangle}:\mathscr{F}^l\rightarrow\mathscr{F}^{l+1}$ are
defined by
$\delta_{\upv}^{\vartriangle}:=\mcI^{\vartriangle}\dv $. Direct computation shows that $\mcI^{\vartriangle}$ is a projection, that is,
\begin{equation}\label{iproj}
\left(\mcI^{\vartriangle}\right)^2=\mcI^{\vartriangle},
\end{equation} and that the conditions for a cochain complex are satisfied by the rows, columns and edge sequence:
\begin{equation}\label{coch}
\mcI^{\vartriangle}\upd_{\uph}^{\vartriangle}=0,\qquad \mcE^{\vartriangle}\upd_{\uph}^{\vartriangle}=0,\qquad \delta_{\upv}^{\vartriangle}\mcE^{\vartriangle}=0,\qquad \left(\delta_{\upv}^{\vartriangle}\right)^2=0.
\end{equation}
Indeed, the augmented difference variational bicomplex is exact, just as in the differential case. A proof of this is outlined in the Appendix.

Let $(\p_{n^1},\p_{n^2},\ldots,\p_{n^p})$ be the duals to the difference one-forms $(\Delta^1,\Delta^2,\ldots,\Delta^p)$; the duals to the differential one-forms $\dv u^\alpha_{\mbJ}$ are $\p/\p{u^{\alpha}_{\mbJ}}$. These satisfy
\begin{equation}
\p_{n^i}\intp\Delta^j=\delta_i^j,\qquad
\p_{n^i}\intp\dv 
u^{\alpha}_{\mbJ}=0,\qquad
\frac{\partial}{\partial u^{\alpha}_{\mbJ}}\intp \Delta^j=0,\qquad \frac{\partial}{\partial u^{\alpha}_{\mbJ}}\intp \dv u^\beta_\mbK=\delta^\beta_\alpha\delta^\mbJ_\mbK.
\end{equation}

For difference equations, the base space $\mathbb{Z}^p$ is discrete, so every (tangent) vector field is vertical. A locally smooth vector field $\mbv_0=Q^\alpha\,\p/\p u^\alpha$ on the total space, prolonged to all orders, yields the vector field $\mbv=\es_{\mbJ}\!Q^\alpha\,\p/\p u^\alpha_{\mbJ}$ on $\mathbb{Z}^p\times P(\mathbb{R}^q)$. In much the same way as for differential equations, this prolongation formula also applies when each $Q^\alpha$ depends on $\mbn$ and finitely many shifts of $\mbu$, in which case $\mbv$ is a generalized vector field on $\mathbb{Z}^p\times P(\mathbb{R}^q)$; moreover,  $\mbv$ commutes with each $\es_i$ \cite{Hy2014}. The $q$-tuple $(Q^1,Q^2,\dots,Q^q)$ that determines the generalized vector field $\mbv$ is called its \textit{characteristic}.

The Lie derivative of a $(k,l)$-form $\sigma$ with respect to a generalized vector field $\mbv$ on $\mathbb{Z}^p\times P(\mathbb{R}^q)$ is, by Cartan's formula,
\begin{equation}
\mcL_\mbv\sigma=\mbv\intp\dv \;\!\!\sigma+\dv\,(\mbv\intp\sigma).
\end{equation}
(See also the definition \eqref{lieddiscrete} through the corresponding transformation.)
This mirrors the differential case, as do the proofs of the following results.
\begin{prop}\label{prop:symmetry}
	Let $\sigma$ be a $(k,l)$-form on $\mathbb{Z}^p\times P(\mathbb{R}^q)$. If $\mbv$ is a generalized vector field on $\mathbb{Z}^p\times P(\mathbb{R}^q)$ then
	\[
	\mbv\intp\dDh\;\!\!\sigma+\dDh(\mbv\intp\sigma)=0,
	\]
	so
	\[
	\mcL_\mbv \sigma=\mbv\intp\dD\;\!\!\sigma+\dD(\mbv\intp\sigma).	
	\]
	Furthermore,
	\[
	\p_{n^i}\intp\dv \;\!\!\sigma+\dv\,(\p_{n^i}\intp\sigma)=0,
	\]
	and the Lie difference of $\sigma$ with respect to the horizontal translation $\te_{\mathbf{1}_i}$ satisfies the identity
	\[
	\diffi\sigma=\p_{n^i}\intp\dDh\;\!\!\sigma+\dDh(\p_{n^i}\intp\sigma).
	\]
	Therefore,
	\[
	\diffi\sigma=\p_{n^i}\intp\dD\;\!\!\sigma+\dD(\p_{n^i}\intp\sigma).
	\]
\end{prop}

Remarkably, both the (vertical) Lie derivative and (horizontal) Lie difference satisfy a formula that is similar to Cartan's, with the exterior difference-derivative on $\mathbb{Z}^p\times P(\mathbb{R}^q)$ replacing the exterior derivative on $J^{\infty}(X\times U)$.

Let
$\msL[\mbu]=L(\mbn,[\mbu])\text{vol}$ be a given Lagrangian difference form, with Lagrangian $L(\mbn,[\mbu])$.
A generalized vector field $\mbv$ on $\mathbb{Z}^p\times P(\mathbb{R}^q)$ is a
\textit{variational symmetry generator} if $\mbv(L)$ is a null Lagrangian, that is, if
there exist functions $F^i(\mbn,[\mbu])$ such that
\begin{equation}
\mbv(L)=\diffi F^i(\mbn,[\mbu]).
\end{equation}
Equivalently, $\mbv$ is a variational symmetry generator if there exists
$\sigma(\mbv)\in\Omega^{p-1,0}$ such that
\begin{equation}\label{eq:vs}
\mbv\intp\dv \msL
=\dDh\!\sigma(\mbv);
\end{equation}
in coordinates,
\begin{equation}
\sigma(\mbv) =F^i(\mbn,[\mbu])\,\p_{n^i}\intp\text{vol}.
\end{equation}

As $\mcI^{\vartriangle}$ is a projection and the difference variational bicomplex is exact, a difference version of equality \eqref{eq:ii} holds. For each $\sigma\in\Omega^{p,l}$, there exists $\tau\in\Omega^{p-1,l}$ such that 
\begin{equation}\label{eq:dissig}
\sigma=\mcI^{\vartriangle}(\sigma)-\upd_{\uph}^{\vartriangle}\!\tau.
\end{equation}
In particular, for $\sigma=\dv \msL$ there exists $\eta\in\Omega^{p-1,1}$ such that
\begin{equation}\label{eq:dvLag}
\dv \msL=\mcE^{\vartriangle}(\msL)
-\dDh\!\eta.
\end{equation}
By using \eqref{eq:vs}, \eqref{eq:dvLag} and Proposition \ref{prop:symmetry}, we obtain
\begin{equation}\label{eq:ddcl}
\dDh\!\left(\sigma(\mbv)-\mbv\intp\eta\right)
=\mbv\intp\mcE^{\vartriangle}(\msL),
\end{equation}
which gives a conservation law \label{sectiondvb}
\begin{equation}
\label{eq:Noecl}
\dDh\!\left(\sigma(\mbv)-\mbv\intp\eta\right)=0
\text{~~on solutions of~~}
\mcE^{\vartriangle}(\msL)=0.
\end{equation}
The conservation law \eqref{eq:Noecl} is a coordinate-free version of the difference conservation law obtained by Noether's (First) Theorem; its differential counterpart was proved in \cite{BrHyLa2010}. (See \cite{Ko2011} for a comprehensive history of Noether's theorems on variational symmetries.)

\section{Discrete mechanics via the difference variational bicomplex}
\label{sec:hasy}
In \cite{MaWe2001} and references therein, discrete mechanics is
developed using the standard approach in classical mechanics, that is,
by studying the discrete equations of motion on a manifold equipped with a
closed nondegenerate two-form. In this section, we apply the augmented difference  variational
bicomplex to discrete mechanics, with base space $\mathbb{Z}$ and the fibre (in total space) $\mathbb{R}^2$ (for simplicity). In the usual mechanics notation\footnote{Throughout this section only, $p$ and $q$ are real-valued variables, not dimensions (which are given).}, let $(n,q,p)$
be the coordinates on the total space $\mathbb{Z}\times \mathbb{R}^2$; let $\es$ be the forward shift in $n$ and the forward difference operator be $\diffn=\es-\id$.

Consider the following $(0,2)$-form, which is vertically-closed and nondegenerate:
\begin{equation}
\omega =\dv p\wedge \dv q.
\end{equation}
This gives each fibre in the total space the structure of a symplectic manifold. Suppose that
the horizontal translation map $\te_1:(n,p,q)\mapsto(n+1,p,q)$ is a symplectomorphism, so that $\te_1^*\omega_{n+1}=\omega_n$. In the prolongation space $\mathbb{Z}\times P(\mathbb{R}^2)$, this condition amounts to $\diffn\omega=0$, that is,
\[
\dv p_1\wedge \dv q_1-\dv p_0\wedge \dv q_0=0.
\]
As the augmented difference  variational bicomplex is exact, there exists a Hamiltonian function $H$ on $\mathbb{Z}\times P(\mathbb{R}^2)$ that satisfies\footnote{One may alternatively choose the Hamiltonian such that
\begin{equation*}
    (p_1-p_0) \dv q_1-(q_1-q_0) \dv p_0 =-\dv H(n,q_1,p_0),
\end{equation*}
and the corresponding symplectic map becomes the Euler-A discretization method.
}
\begin{equation}
(p_1-p_0) \dv q_0-(q_1-q_0) \dv p_1 =-\dv H.
\end{equation}
Consequently, $H$ is a function of $(n,q_0,p_1)$ only. In coordinates, the
symplectic map is
\begin{equation}
\label{eq:flow}
q_1-q_0=\frac{\partial H(n,q_0,p_1)}{\partial p_1},\qquad
p_1-p_0=-\,\frac{\partial H(n,q_0,p_1)}{\partial q_0}.
\end{equation}
With the step-length incorporated into $H$, this is the Euler-B discretization method for a continuous
Hamiltonian system; see \cite{LeRe2004}.

Reversing the above argument (with $p_{-1}$ replacing $p_0$), one can start with a Hamiltonian $H(n,p,q)$ defined on the total space and apply the map
\begin{equation}
\label{eq:flow1}
q_1-q_0=\frac{\partial H(n,q_0,p_0)}{\partial p_0},\qquad
p_0-p_{-1}=-\,\frac{\partial H(n,q_0,p_0)}{\partial q_0}
\end{equation}
on each $P_n(\mathbb{R}^2)$. Then
\[
(p_0-p_{-1}) \dv q_0-(q_1-q_{0}) \dv p_0 =-\dv H,
\]
so the map preserves the symplectic $(0,2)$-form
\begin{equation}\label{om1std}
\omega =\dv p_{-1}\wedge \dv q_0
\end{equation}
on $\mathbb{Z}\times P(\mathbb{R}^2)$. This approach has the advantage that
(similarly to the corresponding continuous case) the symplectic map \eqref{eq:flow1} can be written in terms of a self-adjoint matrix operator:
\begin{equation}
    \left(
\begin{array}{cc}
  0&-(\id-\es^{-1})\\
  \es-\id&0
\end{array}
    \right)
      \left(
\begin{array}{c}
  q_0\\
  p_0
\end{array}
    \right)=
         \left(
    \begin{array}{c}
    \partial H(n,q_0,p_0)/\partial q_0\\
    \partial H(n,q_0,p_0)/\partial p_0
    \end{array}
    \right).
\end{equation}

The system \eqref{eq:flow1} amounts to the Euler--Lagrange equations for the Lagrangian $(1,0)$-form
\[
\msL=\{p_0(q_1-q_0)-H(n,q_0,p_0)\}\Delta_n.
\]
Specifically,
\[
\mcE^{\vartriangle}(\msL)=\dv \msL+\upd_{\uph}^{\vartriangle}\!\eta,\qquad \eta=p_{-1}\dv q_0.
\]
So $\omega=\dv \eta$.

For symplectic difference maps in general, at most one of the Hamiltonian function and the symplectic form is defined on the total space, so it is essential to work in an appropriate prolongation space.

So far, we have worked mainly in terms of the given coordinates. For an entirely coordinate-free formulation, the exterior difference operator $\dDh$ is used in place of the Lie difference $\diffn$. This construction is easily extended to mechanical systems with higher-dimensional fibres. We now generalize it to higher-dimensional base spaces, with application to multisymplectic P$\Delta$Es.
\section{Multisymplectic systems of P$\Delta$Es and the bicomplex}
\label{sec:mupde}

A system of \pde{s} is multisymplectic if there exists a vertically-closed
$(p-1,2)$-form, $\omega$, that satisfies the multisymplectic conservation law $\dDh\omega=0$ on all solutions of the system. As $\omega$ is vertically-closed, there exists an $\eta\in\Omega^{p-1,1}$ such that
\begin{equation}
\dv \eta=\omega.
\end{equation}
Consequently, on all solutions of the given
system,
\begin{equation}
\dv \dDh\!\eta=
-\dDh\!\dv \eta=
-\dDh\!\omega=0.
\end{equation}
This implies the existence (locally) of an \textit{associated Lagrangian form},
$\msL\in\Omega^{p,0}$, such that on all solutions,
\begin{equation}\label{etalag}
\dDh\!\eta=-\dv \msL.
\end{equation}
Therefore, as
\begin{equation}
\mcE^{\vartriangle}(\msL)=\mcI^{\vartriangle}(\dv 
\msL)=\mcI^{\vartriangle}(-\dDh\!\eta)=0~~~~(\text{on
solutions}),
\end{equation}
every solution of the system satisfies the Euler--Lagrange equations arising from the associated Lagrangian form $\msL$. However, just as in the continuous case, the Euler--Lagrange equations may have other solutions.

Given a first-order
Lagrangian quasilinear $(p,0)$-form,
\begin{equation}
\begin{aligned}
\label{lp}
\msL&=L_{\beta}^i(\mbn,\mbu)\dDh\!u^{\beta}
\wedge(\partial_{n^i}\intp\volf)-H(\mbn,\mbu)
\volf\\
&=
\left\{L_{\beta}^i(\mbn,\mbu)\left(\diffi u^{\beta}\right)-H(\mbn,\mbu)\right\}
\volf,
\end{aligned}
\end{equation}
the Euler--Lagrange equations are
\begin{equation}
\label{el} 
\frac{\partial L_{\beta}^i(\mbn,\mbu)}{\partial
u^{\alpha}}\diffi u^{\beta}+
\sum_i{(\es_i^{-1}-\id)L_{\alpha}^i(\mbn,\mbu)}-\frac{\partial
H(\mbn,\mbu)}{\partial u^{\alpha}}=0.
\end{equation}
The following $(p-1,1)$-form $\eta$ satisfies
\eqref{eq:dvLag}:
\begin{equation}\label{diseta}
\eta=\sum_{i}\Big(\left(\es_i^{-1}L_{\alpha}^i(\mbn,\mbu)\right)\dv 
u^{\alpha}\wedge(\partial_{n^i}\intp\volf)\Big).
\end{equation}
It leads to the multisymplectic $(p-1,2)$-form,
\begin{equation}\label{olol}
\omega=\dv \eta
=\sum_{i}\left(\left(\es_i^{-1}\frac{\partial
L_{\alpha}^i(\mbn,\mbu)}{\partial u^{\beta}}\right)
\dv\,(\es_i^{-1}\!u^{\beta})
\wedge\dv u^{\alpha}\wedge
(\partial_{n^i}\intp\volf)\right).
\end{equation}
From \eqref{eq:dvLag},
\begin{equation}
    \dDh\omega=-\dv \dDh\eta=-\dv\,(\mcE^{\vartriangle}(\msL)),
\end{equation}
which is zero on all solutions of the Euler--Lagrange equations \eqref{el}. So \eqref{el} is a multisymplectic system of difference equations and $\msL$ is an associated Lagrangian form for $\omega$.

By a similar construction, the Euler--Lagrange equations for any difference Lagrangian constitute a multisymplectic system of \pde{s}. However, we will restrict attention to first-order quasilinear Lagrangian forms \eqref{lp}, as every more complicated Lagrangian $(p,0)$-form can be replaced by an equivalent first-order quasilinear form. For this reason, we regard \eqref{el} as the \textit{standard form} for a multisymplectic system of \pde{s}, with \eqref{diseta} and \eqref{olol} as the standard forms for $\eta$ and $\omega$ respectively.

\section{Multimomentum maps and conservation laws}
\label{sec:mommap}

By definition, every multisymplectic system has a form-valued conservation law, $\dDh\omega=0$ on solutions.
We now define multimomentum  maps for a given multisymplectic
system of \pde{s}; such maps can be used to derive scalar conservation laws. Let $G$ denote a Lie group of point transformations that preserve
the $(p-1,2)$-form $\omega$ on solutions of the system, and let $\mathfrak{g}$ be the associated Lie algebra. For each $\xi\in \mathfrak{g}$, the associated infinitesimal generator is
\begin{equation}
\mbv_{\xi}:=\frac{\upd}{\upd\!\varepsilon}\Big|_{\varepsilon=0}\exp(\varepsilon\xi)(\mbn,\mbu),
\end{equation}
and its characteristic has components $Q^{\alpha}=\mbv_{\xi}\intp\dv u^{\alpha}$. For convenience, we also use $\mbv_{\xi}$ to denote the prolongation of the infinitesimal generator to all orders, so that it gives the infinitesimal action on all forms on $\mathbb{Z}^p\times P(\mathbb{R}^q)$. The preservation of $\omega$ amounts to 
\begin{equation}   \dv\,(\mbv_\xi\intp\omega)=\mcL_{\mbv_\xi}\omega=0\quad\text{on solutions}.
\end{equation}
Indeed, this condition can be solved to find all infinitesimal generators of a given class\footnote{We have restricted attention to Lie point symmetries, for which there is a well-defined Lie group action. However, the calculations extend immediately to generalized symmetries, just as in Noether's Theorem.} that preserve $\omega$. As $\mbv_{\xi}\intp\omega$ is vertically exact, there exists a $\lambda_{\xi}\in\Omega^{p-1,0}$ such
that, on solutions,
\begin{equation}
\label{eq5}
\mbv_{\xi}\intp\omega=\dv \lambda_{\xi}.
\end{equation}
Let $\mathfrak{g}^{\ast}$ be the dual space of the Lie algebra
$\mathfrak{g}$. Then, by analogy with the continuous case, we define a difference multimomentum map
$J:\mathbb{Z}^p\times P(\mathbb{R}^q)\to\mathfrak{g}^{\ast}\otimes\Omega^{p-1,0}$
by
\begin{equation}
\label{mm} J(\mbn,[\mbu])(\xi)=\lambda_{\xi}(\mbn,[\mbu]).
\end{equation}
Now suppose that the $\msL$ is a Lagrangian form for the Euler--Lagrange equations corresponding to $\omega$, and that $\eta$ satisfies \eqref{etalag}. If $\mbv_\xi$ generates variational symmetries, the argument leading to Noether's finite difference conservation law \eqref{eq:Noecl}
gives the difference multimomentum map,
\begin{equation}
\lambda_{\xi}=\sigma(\mbv_{\xi})-\mbv_{\xi}\intp\eta.
\end{equation}
To summarize, we have shown that $\lambda_{\xi}$ is a conservation law if
\begin{equation}
\begin{aligned}
\label{con}
&\dv \lambda_{\xi}=\mbv_{\xi}\intp\omega,\\
&\dDh\!\lambda_{\xi}=\mbv_{\xi}\intp\mcE^{\vartriangle}(\msL),
\end{aligned}
\end{equation}
the second of which is required for $\mbv_{\xi}$ to generate variational symmetries, by \eqref{eq:ddcl}.
\begin{thm}
Given a multisymplectic system whose associated Lagrangian form is $\msL$, any discrete
multimomentum map $\lambda_{\xi}(\mbn,[\mbu])$ satisfying
(\ref{con}) gives rise to a conservation law
$\dDh\!\lambda_{\xi}=0$
on all solutions of the system.
\end{thm}
If $\omega$ and $\eta$ are in standard form then 
the second condition in (\ref{con}) amounts to
\begin{equation}
\label{cll} 
\diffi \lambda_{\xi}^i-Q^{\alpha}
\left(\frac{\partial L_{\beta}^i(\mbn,\mbu)}{\partial
u^{\alpha}}\diffi u^{\beta}
+\sum_i(\es_i^{-1}-\id)L_{\alpha}^i(\mbn,\mbu)-\frac{\partial
H(\mbn,\mbu)}{\partial u^{\alpha}}\right)=0,
\end{equation}
where $\lambda_{\xi}^i$ are the components of
\begin{equation}
\lambda_{\xi}=\lambda_{\xi}^i(\mbn,[\mbu])\partial_{n^i}\intp\volf.
\end{equation}

\begin{exm}\label{exm51new}
To illustrate the above, consider the multisymplectic system of P$\Delta$Es defined on $\mathbb{Z}^2\times P(\mathbb{R}^2)$ by the Lagrangian form
\begin{equation}
    \msL=\left(uu_{1,0}+vv_{1,0}-2uv_{0,1}+\frac{u}{v}\right)\volf,
\end{equation}
namely,
\begin{equation}\label{ex1pdes}
    u_{1,0}+u_{-1,0}-2v_{0,1}+\frac{1}{v}=0,\qquad v_{1,0}+v_{-1,0}-2u_{0,-1}-\frac{u}{v^2}=0.
\end{equation}
These P$\Delta$Es are the components of
\begin{equation}
    \mcE^{\vartriangle}(\msL)=\dv \msL+\dDh\!\left\{2u_{0,-1}\dv v\wedge\Delta^1+(u_{-1,0}\dv u+v_{-1,0}\dv v)\wedge\Delta^2\right\},
\end{equation}
and therefore,
\begin{equation}
    \eta=2u_{0,-1}\dv v\wedge\Delta^1+(u_{-1,0}\dv u+v_{-1,0}\dv v)\wedge\Delta^2.
\end{equation}
Consequently, the following vertically-closed multisymplectic $(2,1)$-form is conserved on all solutions of \eqref{ex1pdes}:
\begin{equation}
    \omega=2\dv u_{0,-1}\wedge\dv v\wedge\Delta^1+(\dv u_{-1,0}\wedge\dv u+\dv v_{-1,0}\wedge\dv v)\wedge\Delta^2.
\end{equation}
To find the multimomentum map(s) admitted by \eqref{ex1pdes}, first note that if $\mbv_\xi$ is a (prolonged) point symmetry generator whose characteristic is $(Q^u(\mbn,u,v),Q^v(\mbn,u,v))$, 
\begin{equation}\label{ex1dvlam}   
    \mbv_\xi\intp\omega=\left(2Q^u_{0,-1}\dv v-2Q^v\dv u_{0,-1}\right)\wedge\Delta^1+\left(Q^u_{-1,0}\dv u-Q^u\dv u_{-1,0}+Q^v_{-1,0}\dv v-Q^v\dv v_{-1,0}\right)\wedge\Delta^2.
\end{equation}
The coefficients of the $(2,1)$-form condition $\dv \,(\mbv_\xi\intp\omega)=0$ constrain the characteristic as follows. The coefficient multiplying $\dv u\wedge\dv u_{0,-1}\wedge\Delta^1$ is $-2\partial Q^v/\partial u$, so $Q^v$ is independent of $u$. Similarly, $Q^u$ is independent of $v$. Taking these results into account, the remaining coefficients yield the conditions
\begin{equation}
    \left(\es_1^{-1}+\id\right)\frac{\p Q^u}{\p u}=0,\qquad\left(\es_1^{-1}+\id\right)\frac{\p Q^v}{\p v}=0,\qquad\es_2^{-1}\!\left(\frac{\p Q^u}{\p u}\right)+\frac{\p Q^v}{\p v}=0.
\end{equation}
These conditions imply that the components of the characteristic are of the form
\begin{equation}
    Q^u=a(\mbn)u+b(\mbn),\qquad Q^v=-a_{0,-1}v+c(\mbn),\qquad \text{subject to } a_{-1,0}=-a,
\end{equation}
for some functions $a,b,c$. From
\eqref{ex1dvlam}, the first condition in \eqref{con} is solved by
\begin{equation}
\begin{aligned}
    \lambda_\xi=&\left\{2a_{0,-1}u_{0,-1}v+2b_{0,-1}v-2cu_{0,-1}\right\}\Delta^1\\&+\left\{-au_{-1,0}u+b_{-1,0}u-bu_{-1,0}+a_{0,-1}v_{-1,0}v+c_{-1,0}v-cv_{-1,0}\right\}\Delta^2.
    \end{aligned}
\end{equation}
The second condition in \eqref{con} amounts to the further constraints
\begin{equation}
    a_{0,-1}=-a,\qquad b=c=0,
\end{equation}
and so $a=k(-1)^{n^1+n^2}$, where $k$ is an arbitrary constant. Consequently, the Lie algebra of infinitesimal generators is one-dimensional. The characteristic
\begin{equation}
\left(Q^u,Q^v\right)=\left((-1)^{n^1+n^2}u,(-1)^{n^1+n^2}v\right)    
\end{equation}
corresponds to the multimomentum map
\begin{equation}
    \lambda_\xi=2(-1)^{n^1+n^2+1}u_{0,-1}v\,\Delta^1+\left\{(-1)^{n^1+n^2+1}\left(u_{-1,0}u+v_{-1,0}v\right)\right\}\Delta^2,
\end{equation}
which gives the conservation law
\begin{equation}
    \diff{1}\!\left\{(-1)^{n^1+n^2+1}\left(u_{-1,0}u+v_{-1,0}v\right)\right\}+\diff{2}\!\left\{2(-1)^{n^1+n^2}u_{0,-1}v\right\}=0.
\end{equation}
\end{exm}
    
\begin{exm}\label{exm52}

This example illustrates how to classify all multimomentum maps admitted by a given class of Lagrangian forms. Consider a multisymplectic P$\Delta$E in the class defined on $\mathbb{Z}^2\times P(\mathbb{R})$ by a first-order quasilinear Lagrangian form, \eqref{lp}. For simplicity, we restrict attention to P$\Delta$Es that have no terms depending on $\mbn$ only. In coordinates, with $\mbn=(n^1,n^2)$,
\begin{equation}
\msL =\left\{L^1(\mbn,u)\diff{1}u+L^2(\mbn,u)\diff{2}u-H(\mbn,u)\right\}\volf,
\end{equation}
where $L^1$ and $L^2$ have no terms that depend on $\mbn$ only. The multisymplectic P$\Delta$E is given by $\mcE^{\vartriangle}(\msL)=0$, where
\begin{align}
\mcE^{\vartriangle}(\msL)=&\left(\frac{\partial L^1}{\partial u}\diff{1}u+\frac{\partial L^2}{\partial u}\diff{2}u-\diff{1}(\es_{1}^{-1}\!L^1)-\diff{2}(\es_{2}^{-1}\!L^2)-\frac{\partial H}{\partial u}\right)\dv u\wedge \volf\\
=& \dv \msL +\dDh\left\{\left(\es_1^{-1}\! L^1\right)\dv u\wedge \Delta^2-\left(\es_1^{-1}\!L^2\right)\dv u\wedge \Delta^1\right\}.\nonumber
\end{align}
Consequently,
\begin{equation}
    \eta=\left(\es_1^{-1}L^1\right)\dv u\wedge \Delta^2-\left(\es_1^{-1}L^2\right)\dv u\wedge \Delta^1,
\end{equation}
and so the multisymplectic $(1,2)$-form is 
\begin{equation}
\omega = \dv \eta=\left(\es_{1}^{-1}\frac{\partial L^1}{\partial u}\right)\dv u_{-1,0}\wedge\dv 
u\wedge\Delta^2- \left(\es_{2}^{-1}\frac{\partial L^2}{\partial u}\right)\dv u_{0,-1}\wedge\dv u
\wedge\Delta^1.
\end{equation}
Let  $\mbv_{\xi}$ be a prolonged vector field whose characteristic is $Q(\mbn,u)$. The condition
$\dv \left(\mbv_{\xi}\intp\omega\right)=0$ amounts to
\begin{equation}\label{ex2pres}
\begin{aligned}
   0&=\es_1^{-1}\!\left(\frac{\partial Q}{\partial u}\frac{\partial L^1}{\partial u}+Q\,\frac{\partial^2 L^1}{\partial u^2}\right)+\frac{\partial Q}{\partial u}\es_1^{-1}\!\left(\frac{\partial L^1}{\partial u}\right),\\
   0&=\es_2^{-1}\!\left(\frac{\partial Q}{\partial u}\frac{\partial L^2}{\partial u}+Q\,\frac{\partial^2 L^2}{\partial u^2}\right)+\frac{\partial Q}{\partial u}\es_2^{-1}\!\left(\frac{\partial L^2}{\partial u}\right).
\end{aligned}
\end{equation}
Consequently, $\partial^2 Q/\partial u^2=0$, so we seek functions $a$ and $b$ such that $Q=a(\mbn)u+b(\mbn)$ is nonzero.

If $a(\mbn)=0$ then $b(\mbn)\neq 0$ and so, from \eqref{ex2pres}, $L^1$ and $L^2$ are linear in $u$. As these components are assumed to have no terms that depend on $\mbn$ only, $L^1=c(\mbn)u$ and $L^2=d(\mbn)u$ for some functions $c$ and $d$. 
Therefore,
\begin{equation}
\mbv_\xi\intp\omega=b_{-1,0}c_{-1,0}\dv u\wedge\Delta^2-bc_{-1,0}\dv u_{-1,0}\wedge\Delta^2-b_{0,-1}d_{0,-1}\dv u\wedge\Delta^1+bd_{0,-1}\dv u_{0,-1}\wedge\Delta^1,
\end{equation}
and hence, from \eqref{con},
\begin{equation}
\lambda_{\xi}=c_{-1,0}(b_{-1,0} u-bu_{-1,0})\Delta^2-d_{0,-1}(b_{0,-1} u-bu_{0,-1})\Delta^1.
\end{equation}
The second condition in \eqref{con} is satisfied only when 
\begin{equation}
H(\mbn,u)=\left(\frac{b_{1,0}c+b_{-1,0}c_{-1,0}+b_{0,1}d+b_{0,-1}d_{0,-1}}{2b}-c-d\right)u^2,
\end{equation}
so the Euler--Lagrange equation is a linear P$\Delta$E that is equivalent to the conservation law
\begin{equation}
\diff{1}\{c_{-1,0}(b_{-1,0} u-bu_{-1,0})\}+\diff{2}\{d_{0,-1}(b_{0,-1} u-bu_{0,-1})\}=0.
\end{equation}

If $a(\mbn)\neq  0$, the solution of \eqref{ex2pres} for which $L^1$ and $L^2$ have no terms that depend on $\mbu$ only is
\begin{equation}
L^1(\mbn,u)=c(\mbn)(au+b)^{-a_{1,0}/a},\qquad L^2(\mbn,u)=d(\mbn)(au+b)^{-a_{0,1}/a},
\end{equation}
where $c$ and $d$ are arbitrary functions. The first condition in \eqref{con} gives the multimomentum map
\begin{equation}\label{ex2nl}
    \lambda_\xi=-c_{-1,0}(a_{-1,0}u_{-1,0}+b_{-1,0})^{-a/a_{-1,0}}(au+b)\Delta^2+d_{0,-1}(a_{0,-1}u_{0,-1}+b_{0,-1})^{-a/a_{0,-1}}(au+b)\Delta^1.
\end{equation}
From the second condition in \eqref{con}, we obtain
\begin{equation}
H(\mbn,u)=-c(au+b)^{-a_{1,0}/a}\left(u+\frac{b_{1,0}}{a_{1,0}}\right)-d(au+b)^{-a_{0,1}/a}\left(u+\frac{b_{0,1}}{a_{0,1}}\right).
\end{equation}
From \eqref{ex2nl}, the nonlinear multisymplectic P$\Delta$E given by the Euler--Lagrange equation satisfies the conservation law
\begin{equation}
\diff{1}\!\left\{-c_{-1,0}(a_{-1,0}u_{-1,0}\!+\!b_{-1,0})^{-a/a_{-1,0}}(au\!+\!b)\right\}+\diff{2}\!\left\{-d_{0,-1}(a_{0,-1}u_{0,-1}\!+\!b_{0,-1})^{-a/a_{0,-1}}(au\!+\!b)\right\}=0.
\end{equation}
\end{exm}

\section{Multisymplectic integrators on logically-rectangular meshes}
\label{sec:multisyin}
In this section, we apply the difference variational bicomplex
on $\mathbb{Z}^p$ to multisymplectic finite difference methods on logically-rectangular meshes that need not be uniform. 

The meshpoints, $\mbx_\mbn$, are assumed to lie on a $p$-dimensional Riemannian manifold and are labelled by $\mbn\in\mathbb{Z}^p$. As before, the dependent variables are $u^\alpha,\ \alpha=1,\dots,q$. The action of the forward shift operators $\es_i$ on the meshpoints is given by $\es_i(\mbx_\mbn)=\mbx_{\mbn+\mathbf{1}_i}$, and so $\es_\mbJ(\mbx_\mbn)=\mbx_{\mbn+\mbJ}$. As before, we work on the prolongation space over a fixed $\mbn$, denoting the value of $u^\alpha$ at $\mbx_{\mbn+\mbJ}$ by $u^\alpha_\mbJ$; for simplicity, we write $u^\alpha$ in place of $u^\alpha_{\mathbf{0}}$. 

If $d(\cdot,\cdot)$ is the distance defined by the Riemannian metric, the step size between adjacent meshpoints, $\mbx_\mbn$ and $\mbx_{\mbn+\mathbf{1}_i}$, is
the length of the local geodesic
connecting them:
\begin{equation}
\epsilon_\mbn^i:=d(\mbx_{\mbn+\mathbf{1}_i},\mbx_{\mbn}).
\end{equation}
In particular, for the Euclidean space
$\mathbb{R}^p$, the step size is $\epsilon_\mbn^i=x^i_{\mbn+\mathbf{1}_i}-x^i_\mbn$. The exterior difference operator can be written as
\begin{equation}
\upd_{\uph}^{\vartriangle}=\sum_{i=1}^p\Delta^i\wedge\diffi=\sum_{i=1}^p\left(\epsilon_\mbn^i\Delta^i\right)\wedge(\epsilon_\mbn^i)^{-1}\diffi.
\end{equation}
When the limit $\epsilon_\mbn^i\rightarrow 0^+$ is taken, $\epsilon_\mbn^i\Delta^i$ tends to
$\upd\!x^i$ and
$(\epsilon_\mbn^i)^{-1}\diffi$ tends to the total
derivative  $D_i$. From this perspective, the horizontal operator $\upd_{\uph}^{\vartriangle}$ may be regarded as an approximation to the horizontal derivative $\upd_{\uph}\!=\upd\!x^i\wedge
D_i$. This makes it straightforward to study multisymplectic integrators using the difference variational bicomplex.

\begin{exm}\label{slbridges}
Consider the following semilinear scalar PDE on $\mbx=(x,t)\in\mathbb{R}^2$:
\begin{equation}\label{semil}
  u_{tt}+\varepsilon u_{xx}+V'(u)=0,\qquad\varepsilon=\pm1;
\end{equation}
here $V(u)$ is a given potential function. Bridges \& Reich \cite{BrRe2001} embedded \eqref{semil} in a first-order multisymplectic system of PDEs,
\begin{equation}
    -v_{t}-w_{x}=V'(u),\qquad u_{t}+p_{x}=v,\qquad
    u_{x}-\varepsilon p_{t}=\varepsilon w,\qquad\varepsilon
    w_{t}-v_{x}=0,
\end{equation}
which they discretized using the following St\"{o}rmer--Verlet scheme (staggered in both the $x$- and
the $t$-directions):
\begin{equation}\label{mi1}
  \begin{aligned}
    -\frac{v_{0,\frac12}-v_{0,-\frac12}}{h_t}-\frac{w_{\frac12,0}-w_{-\frac12,0}}{h_x}&=V'(u),\qquad&
    \frac{u_{0,1}-u}{h_t}+\frac{p_{\frac12,\frac12}-p_{-\frac12,\frac12}}{h_x}&=v_{0,\frac12},\\
    \frac{u_{1,0}-u}{h_x}-\varepsilon\,\frac{p_{\frac12,\frac12}-p_{\frac12,-\frac12}}{h_t}&=\varepsilon w_{\frac12,0},\qquad&
    \varepsilon\,\frac{w_{\frac12,1}-w_{\frac12,0}}{h_t}-\frac{v_{1,\frac12}-v_{0,\frac12}}{h_x}&=0,
  \end{aligned}
\end{equation}
where the step sizes $h_x=\mbx_{1,0}-\mbx$ and $h_t=\mbx_{0,1}-\mbx$ are assumed to be uniform. The centrepoint notation denotes an average, as follows: 
\begin{equation}\label{slvars}
w_{\frac12,0} = \frac{w+w_{1,0}}{2}\,, \qquad v_{0,\frac12} = \frac{v+v_{0,1}}{2}\,, \qquad p_{\frac12,\frac12}=\frac{p+p_{1,0}+p_{0,1}+p_{1,1}}{4}\,.
\end{equation}
The above scheme regards $u$ as being defined on meshpoints, whereas $w$ and $v$ are defined on edges, and $p$ is defined on areas, so $u$ together with the variables in \eqref{slvars} are taken to be the fundamental dependent variables. Shifts of these variables are denoted by $w_{\frac12,1}=\es_1\left(w_{\frac12,0}\right)$, etc. Let $\delta_x=h_x^{-1}\diff{1}$ and $\delta_t=h_t^{-1}\diff{2}$ be the scaled difference operators.
The scaled difference one-forms on the mesh are generated by
\begin{equation}
   \Delta^x=h_x\Delta^1,\qquad\Delta^t=h_t\Delta^2,
\end{equation}
where $\Delta^1$ and $\Delta^2$ are the standard difference one-forms on $\mathbb{Z}^2$.
The multisymplectic system (\ref{mi1}) is the set of Euler--Lagrange equations for the Lagrangian $(2,0)$-form
\begin{equation}
 \msL=\Bigg\{ w_{\frac12,0}\,\delta_x u-p_{\frac12,\frac12}\,\delta_x v_{0,\frac12}+v_{0,\frac12}\,\delta_t u+\varepsilon
 p_{\frac12,\frac12}\,\delta_t w_{\frac12,0}
 -\left(V(u)+\tfrac{1}{2}\left(v_{0,\frac12}\right)^2+\tfrac{\varepsilon}{2}\left(w_{\frac12,0}\right)^2\right)\Bigg\}\Delta^x\wedge\Delta^t.
\end{equation}
Consequently,
\begin{equation}
  \omega=\left(\dv w_{-\frac12,0}\wedge\dv u
  -\dv p_{-\frac12,\frac12}\wedge
  \dv v_{0,\frac12}\right)\wedge\Delta^t-\left(\dv v_{0,-\frac12}\wedge\dv u+\varepsilon
  \dv p_{\frac12,-\frac12}\wedge\dv w_{\frac12,0}\right)\wedge\Delta^x
\end{equation}
is a multisymplectic $(1,2)$-form that satisfies
\begin{equation}
  \upd_{\uph}^{\vartriangle}\!\omega=0~~\text{on solutions
  of (\ref{mi1})}.
\end{equation}

\end{exm}

\begin{exm}
In \cite{Wa2008}, Wang studied the multisymplectic formulation of the Zakharov system
\begin{equation}
  i\phi_{t}+\phi_{xx}+2\phi\psi=0,~~\psi_{tt}-\psi_{xx}+(|\phi|^2)_{xx}=0,
\end{equation}
where $\phi(x,t)$ is complex-valued and $\psi(x,t)$ is real-valued. This can be rewritten as a first-order system of PDEs by introducing several new real-valued variables:
\begin{equation}
  \begin{aligned}
    -&v_{t}+p_{x}=-2u\psi,~\quad u_{t}+q_{x}=-2v\psi,~\quad -u_{x}=-p,~\quad -v_{x}=-q,\\
    &w_{t}=\psi-u^2-v^2,~\quad -\psi_{t}+\varphi_{x}=0,~\quad -w_{x}=-\varphi.
  \end{aligned}
\end{equation}
Wang proposed an Euler box scheme that is a multisymplectic integrator. In the notation from Example \ref{slbridges}, the scheme is
\begin{equation}\label{eulerscheme}
  \begin{aligned}
    -&\delta_t v+\delta_x p=-2u\psi,~\quad 
    \delta_t u_{0,-1}+\delta_x q=-2v\psi,~\quad -\delta_x u_{-1,0}=-p,~\quad -\delta_x v_{-1,0}=-q,\\
    &\delta_t w=\psi-u^2-v^2,~\quad 
    -\delta_t \psi_{0,-1}+\delta_x \varphi_{-1,0}=0,~\quad -\delta_x w=-\varphi.
  \end{aligned}
\end{equation}
These are the Euler--Lagrange equations for the Lagrangian form
\begin{equation}
    \msL=\left\{-u\delta_t v+\psi\delta_t w+u\delta_x p
    +v\delta_x q-\varphi\delta_x w-\left(\tfrac{1}{2}\psi^2
    -\psi u^2 -\psi v^2-\tfrac{1}{2}p^2-\tfrac{1}{2}q^2-\tfrac{1}{2}\varphi^2\right)\right\}\Delta^x\wedge\Delta^t.
\end{equation}
Therefore, the conserved multisymplectic $(1,2)$-form is
\begin{equation}
  \begin{aligned}
    \omega=&\left(-\dv u_{0,-1}\wedge\dv v
    +\dv \psi_{0,-1}\wedge\dv w\right)\wedge\Delta^x\\
    &-\left(\dv u_{-1,0}\wedge\dv p+
    \dv v_{-1,0}\wedge\dv q-
    \dv \varphi_{-1,0}\wedge\dv w\right)\wedge\Delta^t.
  \end{aligned}
\end{equation}
\end{exm}

\section{Concluding remarks}
The difference variational bicomplex provides a geometric framework for the study of finite difference equations. Just as in the continuous case, it can be used to establish the theoretical foundations of multisymplectic systems of P$\Delta$Es. In particular, the coordinate-free version of Noether's Theorem leads to difference multimomentum maps. Although we have restricted attention to Lie point symmetries, the same construction goes through for generalized symmetries. For such symmetries, the characteristic depends not only on $(\mbn,\mbu)$, but also on shifts of $\mbu$.

It is easiest to frame a given multisymplectic system as Euler--Lagrange equations whose Lagrangian form is first-order and quasilinear.
This can be done by introducing extra variables using Lagrange multipliers. Having found the conserved multisymplectic form $\omega$ and any multimomentum maps, one could then use the Euler--Lagrange equations to eliminate the Lagrange multipliers, yielding the results in terms of the original variables. For instance, the Toda-type equation
\begin{equation}\label{tod}
    \frac{1}{u_{1,1}-u}-\frac{1}{u_{-1,1}-u}-\frac{1}{u_{1,-1}-u}+\frac{1}{u_{-1,-1}-u}=0
\end{equation}
is the Euler--Lagrange equation for the Lagrangian
\begin{equation}
    L=\ln\left|\frac{u_{1,0}-u_{0,1}}{u_{1,1}-u}\right| \,.
\end{equation}
By introducing the extra dependent variables $v=u_{1,0}$, $w=u_{1,0}$ and $z=u_{1,1}$, the system can be recast in standard form, with the Lagrangian form
\begin{equation}
    \msL=\left\{\rho(\diff{1}u+u-v)+\sigma(\diff{2}u+u-w)+\tau(\diff{2}v+v-z)+\ln\left|\frac{v-w}{z-u}\right|\right\}\Delta^1\wedge\Delta^2.
\end{equation}
Here $\rho,\sigma$ and $\tau$ are Lagrange multipliers. This leads to the conserved multisymplectic $(2,1)$-form
\begin{equation}\label{omeasy}
    \omega=\dv \rho_{-1,0}\wedge\dv u\wedge\Delta^2-\left(\dv \sigma_{0,-1}\wedge\dv u+\dv \tau_{0,-1}\wedge\dv v\right)\wedge\Delta^1.
\end{equation}
Using the Euler--Lagrange equations to eliminate the Lagrange multipliers, one obtains the conserved multisymplectic form for \eqref{tod}:
\begin{equation}
    \omega=\left\{\frac{\dv u_{-1,1}\wedge\dv u}{(u_{-1,1}-u)^2}-\frac{\dv u_{-1,-1}\wedge\dv u}{(u_{-1,-1}-u)^2}\right\}\wedge\Delta^2-\left\{\frac{\dv u_{1,-1}\wedge\dv u}{(u_{1,-1}-u)^2}+\frac{\dv u_{1,0}\wedge\dv u_{0,-1}}{(u_{1,0}-u_{0,-1})^2}\right\}\wedge\Delta^1.
\end{equation}
A straightforward calculation shows that $\lambda_\xi=\rho_{-1,0}u\Delta^2+(\sigma_{0,-1}u+\tau_{0,-1}v)\Delta^1$ is a multimomentum map for \eqref{omeasy}. Consequently,
\begin{equation}
    \lambda_\xi=\left\{\frac{u}{u_{-1,1}-u}-\frac{u}{u_{-1,-1}-u}\right\}\Delta^2-\left\{\frac{u}{u_{1,-1}-u}+\frac{u_{1,0}}{u_{1,0}-u_{0,-1}}\right\}\Delta^1
\end{equation}
is a multimomentum map for \eqref{tod}. This corresponds to the scaling Lie group whose characteristic is $Q=u$ and leads to the conservation law
\begin{equation}
    \diff{1}\!\left\{\frac{u}{u_{-1,1}-u}-\frac{u}{u_{-1,-1}-u}\right\}+\diff{2}\!\left\{\frac{u}{u_{1,-1}-u}+\frac{u_{1,0}}{u_{1,0}-u_{0,-1}}\right\}=0.
\end{equation}

\section*{Acknowledgements} 
LP would like to thank Timothy Grant, Robert Gray and Jing Ping Wang for valuable discussions.  The authors would like to thank the Isaac Newton Institute for Mathematical Sciences for support and hospitality during the programme {\it `Geometry, Compatibility and Structure Preservation in Computational Differential Equations'}, when some work on this paper was undertaken. This work was partially supported by JSPS KAKENHI grant number JP24K06852, JST-CREST grant number JPMJCR24Q5,  the Fukuzawa Fund and Academic Development Fund of Keio University, and EPSRC grant number EP/R014604/1.

\clearpage

\appendix

\section*{Appendix A: Exactness of the augmented difference variational bicomplex}\renewcommand{\thesection}{A} 

Exactness of the augmented  difference variational bicomplex has been partially  proved in
\cite{HyMa2004,Ku1985}; the following proofs are based on the thesis \cite{Pe2013} by  constructing homotopy operators.  See also \cite{Zh2010} for Zharinov's construction. Local exactness of the augmented difference variational bicomplex in Fig. \ref{fig:mdvb} on
$\mathbb{Z}^p\times P(\mathbb{R}^q)$ is illustrated by the
following three theorems.

\begin{thm}
For each $k=0,1,2,\ldots,p$, the vertical complex
\begin{equation*}
\Omega^{k,0}\stackrel{\dv }{\longrightarrow}
\Omega^{k,1}\stackrel{\dv }{\longrightarrow}
\Omega^{k,2}\stackrel{\dv }{\longrightarrow}
\cdots
\end{equation*}
is exact.
\end{thm}

\begin{proof}
Denote the prolongation of the vertical vector field
$u^{\alpha}\frac{\partial}{\partial u^{\alpha}}$ as
\begin{equation}
\mbv={(\es_{\mbJ}u^{\alpha})\frac{\partial}{\partial
u^{\alpha}_{\mbJ}}} ={u^{\alpha}_{\mbJ}\frac{\partial}{\partial
u^{\alpha}_{\mbJ}}},
\end{equation}
the flow generated by which on $\mathbb{Z}^p\times P(\mathbb{R}^q)$ is a
one-parameter family of diffeomorphisms
\begin{equation}
\exp(\varepsilon \mbv)(\mbn,\mbu)=(\mbn,\ldots,e^{\varepsilon}u^{\alpha}_{\mbJ},\ldots),
\end{equation}
satisfying\footnote{ The mapping $\exp$ is  called the
exponential mapping, and we claim that it is a diffeomorphism here
as the discrete parts can be considered as fixed parameters.}
\begin{equation}
\frac{\upd}{\upd\!\varepsilon}\Big|_{\varepsilon=0}\exp(\varepsilon
\mbv)(\mbn,\mbu)=\mbv.
\end{equation}
 As a consequence, naturally there exist induced push-forward and
pull-back mappings. Take any form $\sigma\in\Omega^{k,l}$. As the
transformation only occurs in the continuous parts, we can define
a derivative of $\sigma$ as\footnote{ In the continuous setting, this derivative
is exactly the Lie derivative, and we also refer to its notation and
call it the Lie derivative accordingly. It satisfies all the
properties that the canonical Lie derivative owns.}
\begin{equation}\label{lieddiscrete}
\mcL_{\mbv}\sigma=\lim_{\varepsilon\rightarrow0}\frac{\exp(\varepsilon \mbv)^{\ast}(\sigma)-\sigma}{\varepsilon}=\frac{\upd}{\upd\!\varepsilon}\Big|_{\varepsilon=0}
{\exp(\varepsilon \mbv)^{\ast}(\sigma)},
\end{equation}
which implies that
\begin{equation}
\frac{\upd}{\upd\!\varepsilon}{\exp(\varepsilon \mbv)^{\ast}(\sigma)}=
\exp(\varepsilon \mbv)^{\ast}(\mcL_\mbv\sigma).
\end{equation}

The
vertical homotopy operators
$h_{\upv}^{k,l}:\Omega^{k,l}\rightarrow\Omega^{k,l-1},l\geq1$ are defined by
\begin{equation}\label{eq:vehomo}
h^{k,l}_{\upv}(\sigma)=\int_0^1{\frac{1}{\lambda}\exp(\ln\lambda
\mbv)^{\ast}(\mbv\intp\sigma) \upd\!\lambda},
\end{equation}
such that
\begin{equation}
\begin{aligned}
\sigma=\dv\,(h_{\upv}^{k,l}(\sigma))
+h_{\upv}^{k,l+1}(\dv \sigma).
\end{aligned}
\end{equation}
The integrand in \eqref{eq:vehomo} is a smooth function of
$\lambda$ at $\lambda=0$, as $l\geq1$. This finishes the proof.
\end{proof}

\begin{thm}
For each $l\geq1$, the augmented horizontal complex
\begin{equation}
0\rightarrow
\Omega^{0,l}\stackrel{\dDh}{\longrightarrow}
\Omega^{1,l}\stackrel{\dDh}{\longrightarrow}
\cdots\stackrel{\dDh}{\longrightarrow}
\Omega^{p-1,l}\stackrel{\dDh}{\longrightarrow}
\Omega^{p,l}\stackrel{\mcI^{\vartriangle}}{\longrightarrow}
\mathscr{F}^l\rightarrow0 
\end{equation}
is exact.
\end{thm}

\begin{proof}
First of all, notice that $\ker
(\mcI^{\vartriangle})=\operatorname{im}
(\dDh)$, which is an
immediate consequence of \eqref{eq:dissig}. The exactness of the
last part is implied by
$\mathscr{F}^l=\mcI^{\vartriangle}(\Omega^{p,l})$.

For any $\sigma\in\Omega^{k,l}$, define the following operators 
\begin{equation}
F_{\alpha}^{\mbJ}(\sigma)=\sum_{\mbI\supset \mbJ}\begin{pmatrix}
\mbI\\
\mbJ
\end{pmatrix}\es_{-\mbI}\left(\partial_{u_{\mbI}^{\alpha}}\intp\sigma\right).
\end{equation}
Here if we denote $\mbI=(i^1,i^2,\ldots,i^p)$ and $\mbJ=(j^1,j^2,\ldots,j^p)$,
then $\mbI\supset \mbJ$ means that $i^k\geq j^k$ for each $k$, and
\begin{equation}
\begin{pmatrix}
\mbI\\
\mbJ
\end{pmatrix}=\frac{\mbI!}{\mbJ!(\mbI- \mbJ)!}
\end{equation}
 with $\mbI!=i^1!i^2!\ldots i^p!$. 
For any $\sigma\in\Omega^{k,l}$ with $1\leq k\leq p$, we define
the horizontal homotopy operators as
\begin{equation}
h_{\uph}^{k,l}(\sigma)=\frac{1}{l}\sum_{\alpha,m,\mbI}\frac{|i^m|+1}{p-k+|\mbI|+1}(S-\id)_{\mbI}
\left(\dv u^{\alpha}\wedge
F_{\alpha}^{\mbI+\mathbf{1}_m}(\partial_{n^m}\intp\sigma)\right),
\end{equation}
which satisfy
\begin{equation}\label{eq:hohomo}
h_{\uph}^{k+1,l}(\dDh\!\sigma)
+\dDh\left(h_{\uph}^{k,l}(\sigma)\right)=\sigma.
\end{equation} 
When $k=0$, by defining $\Omega^{-1,l}=0$, i.e.,
$\partial_{n^m}\intp\sigma=0$ for each $m$,  the identity \eqref{eq:hohomo} still holds.
\end{proof}

\begin{thm}
The boundary complex\footnote{Note that Kupershmidt in \cite{Ku1985} proved the exactness around $\mcE^{\vartriangle}$. In \cite{HyMa2004}, Hydon \& Mansfield proved the exactness in the difference variational complex. }
\begin{equation}
0\rightarrow\mathbb{R}\rightarrow
\Omega^{0,0}\stackrel{\dDh}{\longrightarrow}
\Omega^{1,0}\stackrel{\dDh}{\longrightarrow}
\cdots\stackrel{\dDh}{\longrightarrow}
\Omega^{p-1,0}\stackrel{\dDh}{\longrightarrow}
\Omega^{p,0}\stackrel{\mcE^{\vartriangle}}{\longrightarrow}
\mathscr{F}^1\stackrel{\delta^{\vartriangle}_{\upv}}{\longrightarrow}
\mathscr{F}^2\stackrel{\delta^{\vartriangle}_{\upv}}{\longrightarrow}
\cdots
\end{equation}
is exact.
\end{thm}

\begin{proof}

Let us break the complex into two pieces
\begin{equation}
0\rightarrow\mathbb{R}\rightarrow
\Omega^{0,0}\stackrel{\dDh}{\longrightarrow}
\Omega^{1,0}\stackrel{\dDh}{\longrightarrow}
\cdots\stackrel{\dDh}{\longrightarrow}
\Omega^{p-1,0}\stackrel{\dDh}{\longrightarrow}
\Omega^{p,0}\stackrel{\mcE^{\vartriangle}}{\longrightarrow}
\mathscr{F}^1 
\end{equation}
and
\begin{equation}
\Omega^{p,0}\stackrel{\mcE^{\vartriangle}}{\longrightarrow}
\mathscr{F}^1\stackrel{\delta^{\vartriangle}_{\upv}}{\longrightarrow}
\mathscr{F}^2\stackrel{\delta^{\vartriangle}_{\upv}}{\longrightarrow}
\cdots. 
\end{equation}

We only show the exactness of the second piece here. For any
$\sigma\in\mathscr{F}^l,l\geq1$, the vertical homotopy
operators give that 
\begin{equation}\label{eq:vehomoeq}
\sigma=\dv\,(h_{\upv}^{p,l}(\sigma))
+h_{\upv}^{p,l+1}(\dv \sigma).
\end{equation}
From \eqref{eq:dissig}, we have
\begin{equation}
\dv \sigma=\delta^{\vartriangle}_{\upv}\sigma
+\dDh\!\tau_1
\end{equation}
and
\begin{equation}
\dv\,(h_{\upv}^{p,l}(\sigma))=\delta^{\vartriangle}_{\upv}
(h_{\upv}^{p,l}(\sigma))+\dDh\!\tau_2,
\end{equation}
for some $\tau_1\in\Omega^{p-1,l+1}$ and $\tau_2\in\Omega^{p-1,l}$.
Since $\mcI^{\vartriangle}(\sigma)=\sigma$, we apply the
difference interior Euler operator $\mcI^{\vartriangle}$ to
equality \eqref{eq:vehomoeq} and obtain 
\begin{equation}\label{eq:Euleromega}
\begin{aligned}
\sigma&=\mcI^{\vartriangle}\left(\delta^{\vartriangle}_{\upv}
(h_{\upv}^{p,l}(\sigma))+\dDh\!\tau_2\right)+
\mcI^{\vartriangle}\circ
h_{\upv}^{p,l+1}\left(\delta^{\vartriangle}_{\upv}\sigma
+\dDh\!\tau_1\right)\\
&=\mcI^{\vartriangle}\circ\delta^{\vartriangle}_{\upv}\circ
h_{\upv}^{p,l}(\sigma) +\mcI^{\vartriangle}\circ
h_{\upv}^{p,l+1}\circ\delta^{\vartriangle}_{\upv}(\sigma)
+\mcI^{\vartriangle}\circ h_{\upv}^{p,l+1}(\dDh\!\tau_2)\\
&=\delta^{\vartriangle}_{\upv}\circ
h_{\upv}^{p,l}(\sigma) +\mcI^{\vartriangle}\circ
h_{\upv}^{p,l+1}\circ\delta^{\vartriangle}_{\upv}(\sigma)
-\mcI^{\vartriangle}\left(\dDh\left( h_{\upv}^{p,l+1}(\tau_2)\right)\right)\\
&=\delta^{\vartriangle}_{\upv}\circ
h_{\upv}^{p,l}(\sigma) +\mcI^{\vartriangle}\circ
h_{\upv}^{p,l+1}\circ\delta^{\vartriangle}_{\upv}(\sigma),
\end{aligned}
\end{equation}
where the third equality holds as
$\dDh$ anti-commutes
with $h_{\upv}^{p,l}$ as consequence of Proposition
\ref{prop:symmetry}. Note that when the operator
$\delta^{\vartriangle}_{\upv}$ is applied to
$(p,0)$-forms, it should be replaced by the difference
Euler--Lagrange operator $\mcE^{\vartriangle}$. To continue
 our proof, the equality \eqref{eq:Euleromega} should be separately
considered through the following two
cases.\\
$i)$ Let $l=1$. Recall that $\delta^{\vartriangle}_{\upv}\mcE^{\vartriangle}=0$, and for any $(p,1)$-form $\sigma$ satisfying
$\delta^{\vartriangle}_{\upv}(\sigma)=0$, \eqref{eq:Euleromega}
leads to that
\begin{equation}
\sigma=\mcE^{\vartriangle}\circ
h_{\upv}^{p,1}(\sigma),
\end{equation}
which implies the exactness. Namely, there exists a $(p,0)$-form
$\tau=h_{\upv}^{p,1}(\sigma)$,
such that $\sigma=\mcE^{\vartriangle}(\tau)$.\\
$ii)$ Let $l\geq2$.  We define homotopy operators to verify the exactness. The
following property is needed that for any $(p,m)$-form $\sigma$,
with $m\geq1$,
\begin{equation}
\delta^{\vartriangle}_{\upv}(\sigma)=
\delta^{\vartriangle}_{\upv}\mcI^{\vartriangle}(\sigma),
\end{equation}
which is an immediate result by applying
$\delta^{\vartriangle}_{\upv}$ to the equality
$\sigma=\mcI^{\vartriangle}(\sigma)+\dDh\!\tau$,
for some $\tau\in\Omega^{p-1,m}$. Therefore, for any
$\sigma\in\mathscr{F}^l,l\geq2$,
$h_{\upv}^{p,l}(\sigma)$ is a $(p,m)$-form with
$m=l-1\geq1$. Now the equality \eqref{eq:Euleromega} becomes 
\begin{equation}
\begin{aligned}
\sigma&=\delta^{\vartriangle}_{\upv}\circ
h_{\upv}^{p,l}(\sigma)
+\mcI^{\vartriangle}\circ h_{\upv}^{p,l+1}\circ\delta^{\vartriangle}_{\upv}(\sigma)\\
&=\delta^{\vartriangle}_{\upv}\circ\mcI^{\vartriangle}\circ
h_{\upv}^{p,l}(\sigma)
+\mcI^{\vartriangle}\circ h_{\upv}^{p,l+1}\circ\delta^{\vartriangle}_{\upv}(\sigma)\\
&=\delta^{\vartriangle}_{\upv}\left(\mcH^l(\sigma)\right)
+\mcH^{l+1}\left(\delta^{\vartriangle}_{\upv}(\sigma)\right),
\end{aligned}
\end{equation}
where the homotopy operators
$\mcH^l:\mathscr{F}^{l}\rightarrow\mathscr{F}^{l-1},l\geq2$
are defined by $\mcH^l=\mcI^{\vartriangle}\circ
h_{\upv}^{p,l}$. This finishes the proof of exactness
for the second piece.
\end{proof}


\end{document}